\documentclass[10pt,twocolumn,twoside,english]{IEEEtran}  % Comment this line out
\IEEEoverridecommandlockouts                              % This command is only
\usepackage{color}
\usepackage{graphicx}
\usepackage{epstopdf}
\usepackage{hyperref}
\usepackage[nocompress]{cite}
\usepackage{amsmath,amsfonts, amssymb}
\newtheorem{problem}{Problem}
\newtheorem{theorem}{Theorem}
\newtheorem{corollary}{Corollary}
\newtheorem{lemma}{Lemma}
\newtheorem{remark}{Remark}
\newtheorem{definition}{Definition}
\newtheorem{example}{Example}
\newtheorem{algorithm}{Algorithm}
\newtheorem{proposition}{Proposition}
\usepackage{epsfig} % for postscript graphics files
\usepackage{psfrag}

\usepackage{times}
\usepackage[tight,footnotesize]{subfigure}
\usepackage{verbatim}
\usepackage{amsmath}
\usepackage{babel}
\usepackage{listings}
\usepackage{comment} 
\usepackage{multirow}
\usepackage{algorithmic}
\usepackage{algorithm}

\usepackage[tight]{subfigure}
\allowdisplaybreaks
\begin{document}
\title{\LARGE \bf Data-Driven Robust Taxi Dispatch under Demand Uncertainties}

\author{Fei Miao,\ \and Shuo Han,\ \and Shan Lin,\ \and Qian Wang,\ \and John Stankovic,\ \and Abdeltawab Hendawi,\ \and Desheng Zhang,\ \and Tian He,\ \and George J. Pappas% <-this % stops a space
\thanks{This work was supported by NSF CPS-1239152, Project Title: CPS: Synergy: Collaborative Research: Multiple-Level Predictive Control of Mobile Cyber Physical Systems with Correlated Context, NSF (CNS-1239224) and TerraSwarm. %, one of six centers of STARnet, a Semiconductor Research Corporation program sponsored by MARCO and DARPA. 
Part of the results of this work appeared at the 54th IEEE Conference on Decision and Control, Osaka, Japan, December 2015~\cite{taxi_cdc2015}. }% <-this % stops a space
\thanks{F.~Miao is with the Department of Computer Science and Engineering, University of Connecticut, Storrs, CT, USA 06269. Email: fei.miao@uconn.edu. S.~Han is with the Department of Electrical and Computer Engineering, University of Illinois at Chicago, Chicago, Illinois, USA 60607. Email: hanshuo@uic.edu. S.~Lin is with Department of Electrical and Computer Engineering, Stony Brook University, Long Island, NY, USA 11794. Email: shan.x.lin@stonybrook.edu. Q.~Wang is with ADVANCE.AI, Beijing, China. Email: qianwangchina@gmail.com. J.~Stankovic and A.~Hendawi are with the Department of Computer Science, University of Virginia, Charlottesville, VA,  USA, 22904. Email: \{stankovic, hendawi\}@virginia.edu.  D.~Zhang is with the Department of Computer Science at Rutgers University, NJ, USA 08854. Email: desheng.zhang@cs.rutgers.edu. T.~He is with Department of Computer Science and Engineering, University of Minnesota, Minneapolis, MN 55455, USA. Email: tianhe@cs.umn.edu. G.~J.~Pappas is with the Department of Electrical and Systems Engineering, University of Pennsylvania, Philadelphia, PA, USA 19014. Email: pappasg@seas.upenn.edu.}
}

\maketitle

\begin{abstract}
In modern taxi networks, large amounts of taxi occupancy status and location data are collected from networked in-vehicle sensors in real-time. They provide knowledge of system models on passenger demand and mobility patterns for efficient taxi dispatch and coordination strategies. Such approaches face new challenges: how to deal with uncertainties of predicted customer demand  while fulfilling the system's performance requirements, including minimizing taxis' total idle mileage and maintaining service fairness across the whole city; how to formulate a computationally tractable problem. To address this problem, we develop a data-driven robust taxi dispatch framework to consider spatial-temporally correlated demand uncertainties. The robust vehicle dispatch problem we formulate is concave in the uncertain demand and convex in the decision variables. Uncertainty sets of random demand vectors are constructed from data based on theories in hypothesis testing, and provide a desired probabilistic guarantee level for the performance of robust taxi dispatch solutions. We prove equivalent computationally tractable forms of the robust dispatch problem using the minimax theorem and strong duality. Evaluations on four years of taxi trip data for New York City show that by selecting a probabilistic guarantee level at 75\%, the average demand-supply ratio error is reduced by 31.7\%, and the average total idle driving distance is reduced by 10.13\% or about 20 million miles annually, compared with non-robust dispatch solutions.
\end{abstract}

\section{Introduction}
\label{sec:intro}
Modern transportation systems are equipped with various sensing technologies for passenger and vehicle tracking, such as radio-frequency identification (RFID) and global positioning system (GPS). Sensing data collected from transportation systems provides us opportunities for understanding spatial-temporal patterns of passenger demand. Methods of predicting taxi-passenger demand~\cite{Desheng_Dmodel, predict_demand}, travel time~\cite{time_bayen, velocity_bayen, rt_trip} and traveling speed~\cite{limit_tt, Jaillet} according to traffic monitoring data have been developed. %Simple travel time predictors were demonstrated to come close to fundamental error bounds in delay prediction~\cite{}. %Models to describe the probability distribution of taxi passenger demand have also been proposed~\cite{predict_demand}.

Based on such rich spatial-temporal information about passenger mobility patterns and demand, many control and coordination solutions have been designed for intelligent transportation systems. Robotic mobility-on-demand systems that minimize the number of re-balancing trips~\cite{mod_balance, mod}, and smart parking systems that allocates resource based on a driver's cost function~\cite{smart_christos} have been proposed. Dispatch algorithms that aim to minimize customers' waiting time~\cite{congest_disp, Lee_review} or to reduce cruising mile~\cite{Desheng_mile} have been developed. In our previous work~\cite{Miao_tase16, taxi_Feiiccps15}, we design a receding horizon control (RHC) framework that incorporates predicted demand model and real-time sensing data. Considering future demand when making the current dispatch decisions helps to reduce autonomous vehicle balancing costs~\cite{mod} and taxis' total idle distance~\cite{taxi_Feiiccps15, Miao_tase16}. Strategies for resource allocation depend on the model of demand in general, and the knowledge and assumptions about the demand affect the performance of the supply-providing approaches~\cite{Arzen_dra},~\cite{Quijano_ra}. These works rely on precise passenger-demand models to make dispatch decisions.

However, passenger-demand models have their intrinsic model uncertainties that result from many factors, such as weather, passenger working schedule, and city events etc. Algorithms that do not consider these uncertainties can lead to inefficient dispatch services, resulting in imbalanced workloads, and increased taxi idle mileage. Although robust optimization aims to minimize the worst-case cost under all possible random parameters, it sacrifices average system performances~\cite{robust_allocate}. For a taxi dispatch system, it is essential to address the trade-off between worst-case and the average dispatch costs under uncertain demand. A promising yet challenging approach is a robust dispatch framework with an uncertain demand model, called an uncertainty set, that captures spatial-temporal correlations of demand uncertainties and the robust optimal solution under this set provides a probabilistic guarantee for the dispatch cost (as defined in problem~\eqref{obj_2}). 

In this work, we consider two aspects of a robust vehicle dispatch model given a taxi-operational records dataset: (1) how to formulate a robust resource allocation problem that dispatches vacant vehicles towards predicted uncertain demand, and (2) how to construct spatial-temporally correlated uncertain demand sets for this robust resource allocation problem without sacrificing too much average performance of the system. We first develop the objective and constraints of a robust dispatch problem considering spatial-temporally correlated demand uncertainties. The objective of a system-level optimal dispatch solution is balancing workload of taxis in each region of the entire city with minimum total current and expected future idle cruising distance. We define an approximation of the balanced vehicle objective in this work, such that the robust vehicle dispatch problem is concave of the uncertain demand and convex of the decision variables. We then design a data-driven algorithm for constructing uncertainty demand sets without assumptions about the true model of the demand vector. The constructing algorithm is based on hypothesis testing theories~\cite{datad_robust}~\cite{N_1970}~\cite{SC_2003}, however, how to apply these theories for spatial-temporally correlated transportation data and uncertainty sets of a robust vehicle resource allocation problem have not been explored before. To the best of our knowledge, this is the first work to design a robust vehicle dispatch model that provides a desired probabilistic guarantee using predictable and realistic demand uncertainty sets.

Furthermore, we explicitly design an algorithm to build demand uncertainty set from data according to different probabilistic guarantee level for the cost. With two types of uncertainty sets --- box type and second-order-cone (SOC) type, we prove equivalent computationally tractable forms of the robust dispatch problem under these uncertainty demand models via the minimax theorem and the strong duality theorem. The robust dispatch problem formulated in this work is convex over the decision variables and concave over the constructed uncertain sets with decision variables on the denominators. This form is not the standard form (i.e., linear programming (LP) or semi-definite programming (SDP) problems) that has already been covered by previous work~\cite{robustconvex, datad_robust, robust_mpc}. With proofs shown in this work, both system performance and computational tractability are guaranteed under spatial-temporal demand uncertainties. The average performance of the robust taxi dispatch solutions with SOC type of uncertain demand set is better compared with that of the box (range) type of uncertainty set in the evaluations based on data. Hence, it is critical to use a more complex type of uncertainty set, the SOC type, and the corresponding robust dispatch model we design in this work. The contributions of this work are:
\begin{itemize}
%\vspace{-6pt}
\item We develop a robust optimization model for taxi dispatch systems under spatial-temporally correlated uncertainties of predicted demand, and define an approximation of the balanced vehicle objective. The robust optimization problem of approximately balancing vacant taxis with least total idle distance is concave of the uncertain demand, convex of the decision variables and computationally tractable under multiple types of uncertainties. 
%\vspace{-6pt}
\item We design a data-driven algorithm to construct uncertainty sets that provide a desired level of probabilistic guarantee for the robust taxi dispatch solutions. 
\item We prove that there exist equivalent computationally tractable convex optimization forms for the robust dispatch problem with both polytope and second-order-cone (SOC) types of  uncertainty sets constructed from data.
\item  Evaluations on four years of taxi trip data in New York City show that the SOC type of uncertain set provides a smaller average dispatch cost than the polytope type. The average demand-supply ratio mismatch is reduced by $31.7\%$, and the average total idle distance is reduced by $10.13\%$ or about $20$ million miles annually with robust dispatch solutions under the SOC type of uncertainty set.
\end{itemize}

The rest of the paper is organized as follows. 
The taxi dispatch problem is described and formulated as a robust optimization problem given a closed and convex uncertainty set in Section~\ref{sec:prob_form}. We design an algorithm for constructing uncertain demand sets based on taxi operational records data in Section~\ref{sec:algorithm1}. Equivalent computationally tractable forms of the robust taxi dispatch problem given different forms of uncertainty sets are proved in Section~\ref{sec:algorithm}. Evaluation results based on a real data set are shown in Section~\ref{sec:simulation}. Concluding remarks are provided in Section~\ref{sec:conclusion}.

\section{Problem Formulation}
\label{sec:prob_form}
%\subsection{Problem Overview}
The goal of taxi dispatch is to direct vacant taxis towards current and predicted future requests with minimum total idle mileage. There are two objectives. One is sending more taxis for more requests to reduce mismatch between supply and demand across all regions in the city. The other is to reduce the total idle driving distance for picking up passengers in order to save cost. Involving predicted future demand when making current decisions benefits to increasing total profits, since drivers are able to travel to regions with better chances to pick up future passengers. In this section, we formulate a taxi dispatch problem with uncertainties in the predicted spatial-temporal patterns of demand. A typical monitoring and dispatch infrastructure is shown in Figure~\ref{sys_structure}. The dispatch center periodically collects and stores real-time information such as GPS location, occupancy status and road conditions; dispatch solutions are sent to taxis via cellular radio. 
\begin{figure}[b!]
\vspace{-18pt}
\centering
\includegraphics [width=0.3\textwidth]{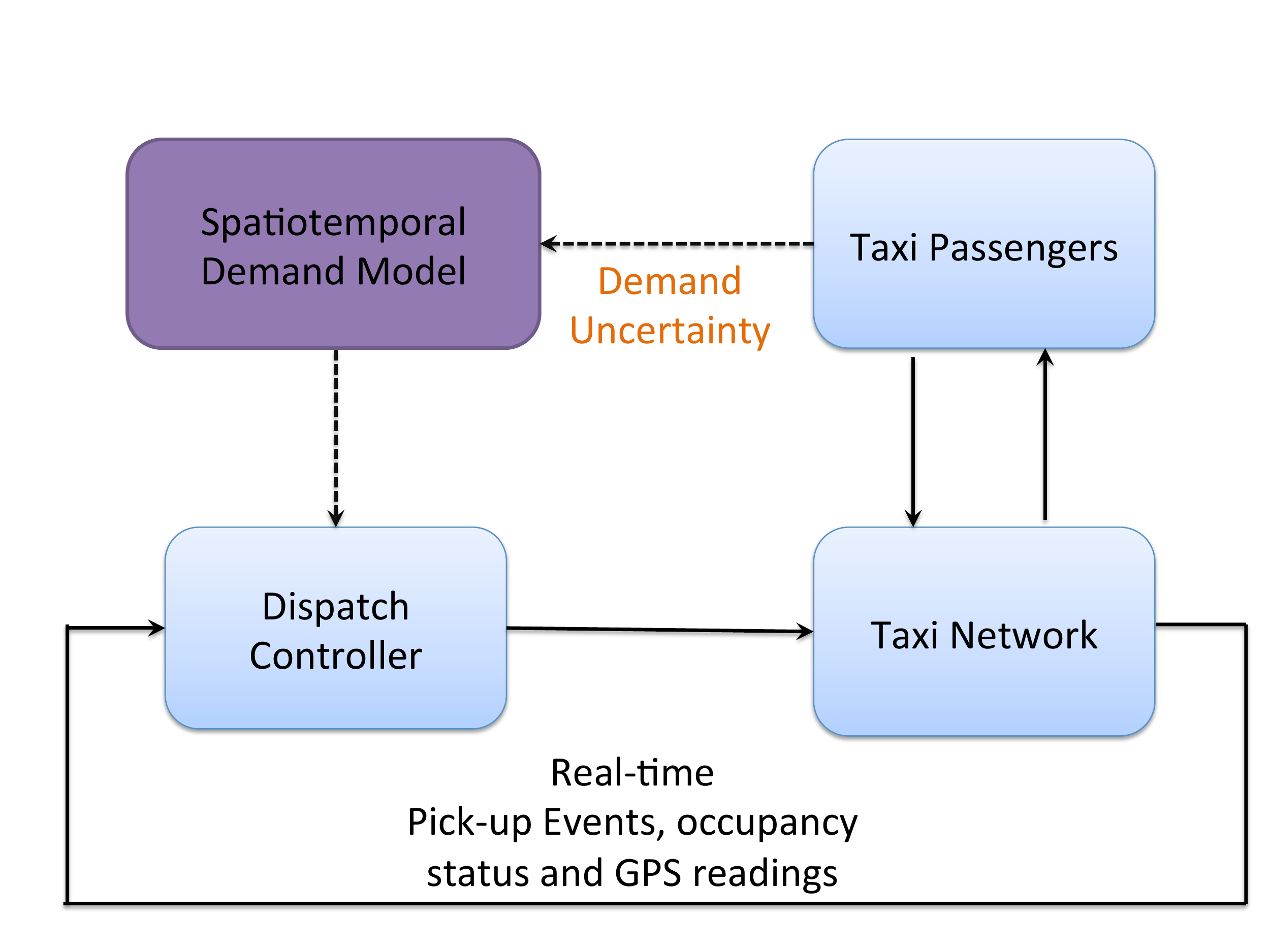}
\vspace{-15pt}
\caption{A prototype of the taxi dispatch system}
\label{sys_structure}
\end{figure}

\subsection{Problem description}
\label{network_model}
\begin{table*}
\centering
%\caption{Parameters of the hybrid stochastic game between the system and the attacker}
\begin{tabular}{|c|c|}
  \hline
  Parameters of~\eqref{Tm}&Description \\ \hline
%  $N$ & the total number of vacant taxis\\ \hline
  $n$ & the number of regions\\ \hline    
  $\tau$ & model predicting time horizon\\ \hline   
   $r^k \in \Delta_k $ & the uncertain total number of requests at each region during time $k$ \\ \hline
   $W \in \mathbb{R}^{n \times n}$ & weight matrix, $W_{ij}$ is the distance from region $i$ to region $j$  \\ \hline
   $P^k \in [0,1]^{n \times n}$ & probability matrix that describes taxi mobility patterns during one time slot\\ \hline %$C^k_{ij}=$ the probability that given a taxi starts at region $i$, it will reach region $j$ by the end of time $k$\\ \hline
   $L^1 \in \mathbb{N}^{n}$ & the initial number of vacant taxis at each region provided by GPS and occupancy status data\\ \hline
   $m \in \mathbb{R}^+$ & the upper bound of distance each taxi can drive idly for picking up a passenger \\ \hline
   %$O^1 \in \mathbb{N}^n$ & the initial number of occupied taxis at each region provided by GPS and occupancy status data\\ \hline
   %$W_i \in \mathbb{R}^{n\times 2}$& preferred positions of the $i$-th taxi at $n$ regions\\\hline
   $\alpha \in \mathbb{R}_+$ & the power on the denominator of the cost function \\ \hline
   $\beta \in \mathbb{R}_+$ & the weight factor of the objective function \\ \hline  
   %$R^k \in \mathbb{R}_+$ & total number of predicted requests in the city\\ \hline                                 
       Variables of~\eqref{Tm}&      \\ \hline    
      $X_{ij}^k \in \mathbb{R}_+$ &  the number of taxis dispatched from region $i$ to region $j$ during time $k$\\ \hline     
      $L^k \in \mathbb{R}^n_+$ & the number of vacant taxis at each region before dispatching at the beginning of time $k$ \\ \hline
       Parameters of Algorithm~\ref{alg: algorithm_1}&          \\ \hline  
$r_c \in \Delta$ & the uncertain concatenated demand vector of $\tau$ consecutive time slots\\ \hline
$\tilde{r}_c(d_l, t, I_p)$ & one sample of $r_c (t)$ according to sub-dataset $I_p$, records of date $d_l$ \\\hline
%$\mathcal{U}_{\epsilon}$ & the uncertainty set that provides $1-\epsilon$ probabilistic guarantee level for problem~\eqref{Tm} \\\hline
$\alpha_h$ & significance level of a hypothesis testing \\ \hline
   % $P^k \in [0, 1]^{N \times n}$ & the predicted region of dispatched taxis at the end of time slot $k$ \\ \hline
    %$d_i^k \in \mathbb{R}_+$ & heuristic idle driving distance of the $i$-th taxi for reaching the suggested location \\ \hline
    %$$
%     $u^k \in \mathbb{R}^{N \times 1}$ & a vector that aggregates $u_i^k, i=1,2,\dots, N$\\ \hline             
\end{tabular}
%\vspace{-8pt}
\caption{Parameters and variables of taxi dispatch problem~\eqref{Tm}.}
\label{T1_parameter}
\vspace{-22pt}
\end{table*} 

For computational efficiency, we assume that the entire city is divided into $n$ regions, and time of one day is discretized to time slots indexed by $t=1, 2, \dots, K$. Taxi dispatch decision is calculated in a receding horizon process, since considering future demand when making the current dispatch decisions helps to reduce resource allocating costs~\cite{mod} and taxis' total idle distance~\cite{Miao_tase16}. At time $t$, we consider the effects of current decision to the following $(t+1,\dots,t+\tau-1)$ time slots. Only the dispatch solution for time $t$ is implemented and solutions for remaining time slots are not materialized. When the time horizon rolls forward by one step from $t$ to $(t+1)$, information about vehicle locations and occupancy status is observed and updated and we calculate a new dispatch solution for $(t+1)$. 
%the uncertainty set $\Delta_k$ non-convex and This type of spatial correlations of demand across each region during the same time slot $k$ is reflected by the correlation of each element of $r^k$.

We define $r^k_j \geqslant 0$ as the number of total requests within region $j$ during time $k$, and $\tau$ is the model predicting time horizon. We relax the integer constraint of $r^k_{j} \in \mathbb{N}$ to positive real, since the integer constraint will make the robust dispatch problem in this section not computationally tractable. The total number of requests at region $j$ may have similar patterns as its neighbors, for instance, during busy hours, several downtown regions may all have peak demand. Meanwhile, demand during several consecutive time slots $r^k$, $k=1,\dots, \tau$ are temporally correlated. Typically, it is difficult to predict a deterministic value of passenger demand of a region during specific time. We define the spatial-temporally correlated uncertain demand by one closed and convex, or compact set $\Delta$ as
\\\centerline{$
r_c=\begin{bmatrix}(r^1)^T, (r^2)^T, \cdots, (r^{\tau})^T \end{bmatrix}^T \in \Delta \subset \mathbb{R}_+^{\tau n}.
 $}               
Where $r_c$ is called the concatenated demand vector, $(r^k)^T$ means the transpose of $r^k$. The closed, bounded, and convex form of $\Delta$ depends on the method to construct the uncertainty set, which we will describe in detail in Section~\ref{sec:algorithm1}. Since $r_c$ depends on $r^k$, and $r^k$ is one component of $r_c$, the uncertainty set for demand $r^k$ at time $k$ is defined as a closed, convex set $\Delta_k$, and a projection of $\Delta$
\\\centerline{$
\Delta_k: = \{r^k\ | \exists r^1,\dots,r^{k-1},r^{k+1},\dots,r^{\tau},\ \text{s.t.}\ r_c \in \Delta\}.
$}
Note that the projection of a convex set onto some of its coordinates is also convex~\cite[Chapter 2.3.2]{book_convex}. %Hence, $\Delta_k$ is also convex, closed and bounded.

%Considered as one type of resource allocation problem, 
A robust dispatch model that decides the amount of vacant taxis sent between each node pair according to the demand at each node and practical constraints is described in a network flow model of Figure~\ref{network}. The edge weight of the graph represents the distance between two regions. Specifically, each region has an initial number of vacant taxis provided by real-time sensing information and an uncertain predicted demand. We define a non-negative decision variable matrix $X^k \in\mathbb{R}_+^{n \times n}$,  $X^k_{ij} \geq 0$, where $X^k_{ij}$ is the number of vehicles dispatched from region $i$ to $j$. We relax the integer constraint of $X^k_{ij} \in \mathbb{N}$ to a non-negative real constraint, since mixed integer programming is not computational tractable with uncertain parameters. Every time when making a resource allocation decision by solving the following robust optimization problem
\begin{align}
\begin{split}
&\underset{X^1}{\text{min}}\ \underset{r^1\in \Delta_1}{\text{max}}\ \underset{X^2}{\text{min}}\ \underset{r^2\in\Delta_2}{\text{max}}\dots\underset{X^{\tau}}{\text{min}}\ \underset{r^{\tau}\in\Delta_{\tau}}{\text{max}}\\
& J = \sum_{k=1}^{\tau} (J_D(X^k)+\beta J_E(X^k,r^k))\quad\quad
\text{s.t.}\quad X^k \in \mathcal{D}_c,
\end{split}
\label{general_Tm}
\end{align}
where $J_D$ is a convex cost function for allocating resources, $J_E$ is a function concave in $r^k$ and convex in $X^k$ that measures the service fairness of the resource allocating strategy, and $\mathcal{D}_c$ is a convex domain of the decision variables that describes the constraints. We define specific formulations of the objective and constraint functions in the rest of this section. 

\begin{figure}[!t]
%\begin{figure}[b!]
\vspace{-8pt}
\centering
\includegraphics [width=0.28\textwidth]{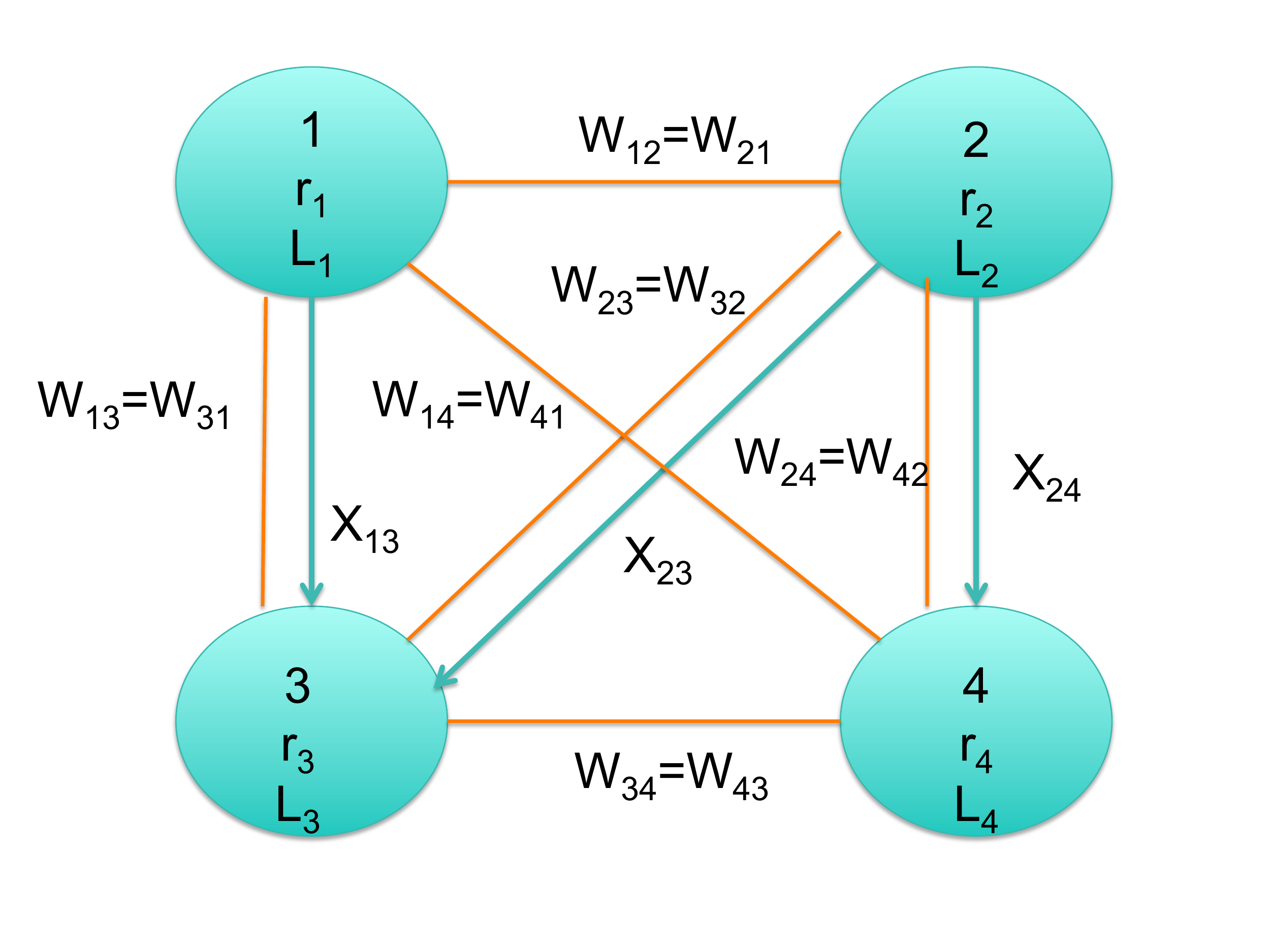}
\vspace{-15pt}
\caption{A network flow model of the robust taxi dispatch problem. A circle represents a region with region ID $1, 2, 3, 4$. We omit the superscript of time $k$ since every parameter is for one time slot only. Uncertain demand is denoted by $r_i$, $L_i$ is the original number of vacant taxis before dispatch at region $i$, and $X_{ij}$ is a dispatch solution that sending the number of vacant taxis from region $i$ to region $j$ with the distance $W_{ij}$.}
% A weight number on the black, dashed edge means the distance between two regions. A red edge with an arrow means sending the corresponding amount of vacant taxis (inside $\{\cdot\}$) to the pointed region.} 
\label{network}
\vspace{-15pt}
\end{figure}
%%%%%%%%%%%%%%%%%%%%%%%%%%%
\subsection{Robust taxi dispatch problem formulation}
\label{dispatch_form}
\textbf{Estimated cross-region idle-driving distance:}
When traversing from region $i$ to region $j$,  taxi drivers take the cost of cruising on the road without picking up a passenger till the target region. Hence, we consider to minimize this kind of idle driving distance while dispatching taxis. We define the weight matrix of the network in Fig.~\ref{network} as $W\in \mathbb{R}^{n\times n}$, where $W_{ij}$ is the distance between region $i$ and region $j$. The across-region idle driving cost according to $X^k$ is 
%To estimate the total distance traveled according to dispatch suggestions, 
\begin{align}
J_D (X^k) = \sum_{i=1}^{n} \sum_{j=1}^n X^k_{ij} W_{ij}.
\label{JD}
\end{align}
We assume that the region division method is time-invariant in this work, and $W$ is a constant matrix for the optimization problem formulation -- for instance, the value of $W_{ij}$ represents the length of shortest path on streets from the center of region $i$ to the center of region $j$
\footnote{For control algorithms with a dynamic region division method, the distance matrix can be generalized to a time dependent matrix $W^k$ as well.}.
%we use the Manhattan norm or one norm between two geometry positions, which is widely applied as a heuristic cost in path planning algorithms~\cite{manhattan}. 

The distance every taxi can drive should be bounded by a threshold parameter $m^k \in \mathbb{R}^+$ during limited time 
%and this distance upper bound for every taxi is denoted . The dispatch solution matrix $X^k$ should satisfy
%\begin{align*}
\centerline{$ 
X^k_{ij}=0\  \text{if}\ W_{ij} > m^k, 
$}
% &X^k_{ij}\geq0, \quad \text{if}\ W_{ij} \leq m,
%\label{bound}
%\end{align*}
which is equivalent to 
 \begin{align}
X_{ij}^k \geqslant 0,\quad X_{ij}^k W_{ij} \leq m^k X^k_{ij},\quad \forall i, j \in \{1, \dots, n\}.
\label{idle_upper}
\end{align}
%\begin{lemma}
To explain this, assume the constraint~\eqref{idle_upper} holds. If $W_{ij} >m^k$ and $X^k_{ij}>0$, we have $X_{ij}^k W_{ij} > m^k X^k_{ij}$, which contradicts to~\eqref{idle_upper}. The threshold $m^k$ is related to the length of time slot and traffic conditions on streets. For instance, with an estimated average speed of cars in one city during time $k=1,\dots,\tau$, and idle driving time to reach a dispatched region is required to be less than $10$ minutes, then the value of $m^k$ should be the distance one taxi can drive during $10$ minutes with the current average speed on road.% (m can also be dependent on $k$, denoted as $m_k$ if a different average speed during each time slot $k$ can be monitored or predicted).  
%Hence, ~\eqref{idle_upper} is the closed form of constraint~\eqref{bound}.
%$X^k_{ij} = 0$ if $W_{ij} > m$. 
%\end{lemma}
%\begin{proof}
%%%%%%%%%%%%%%%%%%%%%%%%%%%%%%%%%%%%
%%%%%%%%%%%%%%%%%%%%%%%%%%%%%%%%%%
\iffalse
We prove by contradiction. Assume is true. Based on the constraint $X^k_{ij} \geq 0, $ if $X^k_{ij} \neq 0$, we must have $X^k_{ij} >0$. When the constraint $X_{ij}^k W_ij \leq m X_{ij}$ is satisfied,  by deriving the left and right side of the inequality by $X_{ij}^k >0$, we have
\begin{align*}
W_{ij} \leq m,
\end{align*}
which is contradicted with $W_{ij} > m$.
\fi
%%%%%%%%%%%%%%%%%%%%%%%%%%%%%%%%%%
%%%%%%%%%%%%%%%%%%%%%%%%%%%
%\end{proof}

%and the longer a time slot is, the larger $m$ can be, because of constraints like speed limit.

\textbf{Metric of serving quality}:
We design the metric of service quality as a function $J_E (X^k, r^k)$ concave in $r^k$ and convex in $X^k$ in this work for computational efficiency~\cite{robustconvex}. Besides vacant taxis traverse to region $j$ according to matrix $X^k$, we define $L^k_j \in \mathbb{R}_+$ as the number of vacant taxis at region $j$ before dispatching at the beginning of time $k$, and $L^1 \in \mathbb{R}_+^{n}$ is provided by real-time sensing information. We assume that the total number of vacant taxis is greater than the number of regions, i.e., $N^k \geqslant n$, and each region should have at least one vacant taxi after dispatch. Then the total number of vacant taxis at region $i$ during time $k$ satisfies that
%\begin{subequations}
\begin{align}
&\sum_{j=1}^{n} X^k_{ji}-\sum_{j=1}^{n} X^k_{ij} + L^k_i >0,\label{total_i}\\
&\sum_{i=1}^{n}\left (\sum_{j=1}^{n} X^k_{ji}-\sum_{j=1}^{n} X^k_{ij} + L^k_i \right)=\sum_{i=1}^n L^k_i =N^k.
\label{total_N}
\end{align}
%\end{subequations}

One service metric is fairness, or that the demand-supply ratio of each region equals to that of the whole city. A balanced distribution of vacant taxis is an indication of good system performance from the perspective that a customer's expected waiting time is short as shown by a queuing theoretic model~\cite{mod}. Meanwhile, a balanced demand-supply ratio means that regions with less demand will get less resources, and idle driving distance will be reduced in regions with more supply than demand if we pre-allocate possible redundant supply to those regions in need. We aim to minimize the mismatch value or the total difference between local region demand-supply ratio and the global demand-supply ratio of the whole city, similarly as the objective defined in~\cite{taxi_Feiiccps15,Miao_tase16}
%\footnotesize
\begin{align}
	\sum_{k=1}^{\tau} \sum_{i=1}^{n}\left|\frac{r^k_i}{\sum\limits_{j=1}^{n} X^k_{ji}-\sum\limits_{j=1}^{n} X^k_{ij} + L^k_i}-\frac{\sum\limits_{j=1}^{n} r_j^k}{N^k}\right|.
	\label{mismatch}
\end{align}
%\normalsize

However, the function~\eqref{mismatch} is not concave in $r^k$ for any $X^k$. It is worth noting we need a function $J_E (X^k, r^k) $ concave in $r^k$ for any $X^k$, and convex in $X^k$ for any $r^k$, to make sure the robust optimization problem is computationally tractable. Hence, we define
\begin{align}
J_E (X^k, r^k) = \sum_{i=1}^{n}\frac{r^k_i}{\left (\sum\limits_{j=1}^{n} X^k_{ji}-\sum\limits_{j=1}^{n} X^k_{ij} + L^k_i \right)^\alpha},\quad\alpha > 0
\label{JE}
\end{align}
as a service fairness metric to minimize. This is because we approximately minimize~\eqref{mismatch} by minimizing~\eqref{JE} under the constraints~\eqref{total_i} and~\eqref{total_N} with an $\alpha$ value chosen according to the desired approximation level, and the following Lemma explains this approximation.
\begin{lemma}
Given deterministic demand vectors $(r^1,\dots,r^{\tau})$ and initial number of vacant vehicles before dispatch $(L^1,\dots, L^{\tau})$ that satisfy $r^k_i \geqslant 1$, $L_i^k \geqslant 0$, $\sum_{i=1}^n L^k_i =N^k$, for any $\epsilon_0 >0$,  any $i\in\{1, \dots, n\}$, $k=1,\dots, \tau$, there exists an $\alpha>0$, such that the optimal solution $(X^k)^*$ by minimizing~\eqref{JE} under constraints~\eqref{total_i} and~\eqref{total_N} satisfies
\begin{align}
\begin{split}
%\left|\frac{r^k_i}{\sum\limits_{j=1}^{n} (X^k_{ji})^*-\sum\limits_{j=1}^{n} (X^k_{ij})^*+ L^k_i}-\frac{\sum\limits_{j=1}^{n} r_j^k}{N^k}\right| < \epsilon_0,\ and\\
\sum\limits_{k=1}^{\tau} \sum_{i=1}^{n}\left|\frac{r^k_i}{\sum\limits_{j=1}^{n} (X^k_{ji})^{*}-\sum\limits_{j=1}^{n} (X^k_{ij})^{*} + L^k_i}-\frac{\sum\limits_{j=1}^{n} r_j^k}{N^k}\right| < n \tau \epsilon_0.
\end{split}
\label{epsilon0_lemma}
\end{align}
\label{jematch}
\end{lemma}
\begin{proof}
See Appendix~\ref{appendix_probform}.
\end{proof}
According to the proof, we can always choose $\alpha$ to be small enough (or close enough to $0$) in order to obtain a desired level of approximation $\epsilon_0$. Hence, in the experiments of Section~\ref{sec:simulation}, we numerically choose $\alpha=0.1$ based on simulation results. Therefore, with function~\eqref{JE}, we map the objective of balancing supply according to demand across every region in the city to a computationally tractable function that concave in the uncertain parameters and convex in the decision variables for a robust optimization problem. 

The number of initial vacant taxis $L_j^{k+1}$ depends on the number of vacant taxis at each region after dispatch during time $k$ and the mobility patterns of passengers during time $k$, while we do not directly control the latter. We define $P^k_{ij}$ as the probability that a taxi traverses from region $i$ to region $j$ and turns vacant again (after one or several drop off events) at the beginning of time $(k+1)$, provided it is vacant at the beginning of $k$. Methods of getting $P^k_{ij}$ based on data include but not limited to modeling trip patterns of taxis~\cite{taxi_Feiiccps15} and autonomous mobility on demand systems~\cite{mod}. Then the number of vacant taxis within each region $j$ by the end of time $k$ satisfies
%We define a function $p^k: \mathbb{R}_+^n \to \mathbb{R}_+^n$ to describe the region transfer pattern of taxis between time $k-1$ and $k$, such that
%\begin{align*}
%\centerline{$
%$\left(\mathbf{1}^T_n X^{k} -(X^{k} \mathbf{1}_n)^T+(L^{k})^T \right) P^k_{\cdot j},$
%$}
%\end{align*}
%where $P^k_{\cdot j}$ is the $j$-th column of $P^k$, and %Vector $L^{k+1}\in\mathbb{R}_+^n$ satisfies
\begin{align}
\vspace{-8pt}
(L^{k+1})^T=( \mathbf{1}^T_n X^k -(X^k \mathbf{1}_n)^T+(L^k)^T )P^k.
%L^k = p^k (X^{k-1}, L^{k-1}).
\label{Lk}
\end{align}
%This equation describes the trip decided by customers, that the number of vacant taxis before dispatching at $k$ given the number of $(X^{k-1}, L^{k-1})$ is not controlled by the system.  

\textbf{Weighted-sum objective function:} Since there exists a trade-off between two objectives, we define a weighted-sum with parameter $\beta > 0$ of the two objectives $J_D(X^k)$ defined in~\eqref{JD} and $J_E(X^k, r^k)$ defined in~\eqref{JE} as the objective function. Let $X^{1:\tau}$ and $L^{2:\tau}$ represent decision variables $(X^1, \dots,X^{\tau})$ and $(L^2,\dots, L^{\tau})$. Without considering model uncertainties corresponding to $r^k$, a convex optimization form of taxi dispatch problem is 
%\footnotesize
\begin{align}
\begin{split}
%&\underset{X^1,\dots, X^{\tau}, L^2,\dots, L^{\tau}}{\text{min.}}
&\underset{X^{1:\tau}, L^{2:\tau}}{\text{min}}
 \quad J =\sum_{k=1}^{\tau} (J_D(X^k)+\beta J_E(X^k,r^k))\\
%&= \sum_{k=1}^{\tau} \sum_{i}\left( \sum_{j} X^k_{ij} W_{ij}+\frac{\beta r^k_i}{(\mathbf{1}_n^T X^k_{\cdot i}-X^k_{i\cdot}\mathbf{1}_n+ L^k_i)^{\alpha}}\right) \\
&\text{s.t.\ \     ~\eqref{idle_upper},~\eqref{total_i},~\eqref{Lk}}.
%\quad (L^{k+1})^T=( \mathbf{1}^T_n X^k -X^k \mathbf{1}_n+(L^k)^T )P^k, \\
%&\quad\quad \mathbf{1}^T_n X^k -X^k \mathbf{1}_n+(L^k)^T >0,\\
%&\quad\quad  X^k_{ij} W_{ij} \leq m X^k_{ij}, \\
%&\quad\quad X^k_{ij} \geq 0, i,  j \in \{1, 2, \dots, n \}.
\end{split}
\label{nonr}
%\vspace{-8pt}
\end{align}
%\normalsize
%and keep all the constraints the same, we have a convex optimization problem. 

\textbf{Robust taxi dispatch problem formulation}:
We aim to find out a dispatch solution robust to an uncertain demand model in this work. For time $k=1,\dots, \tau$, uncertain demand $r^k$ only affects the dispatch solutions of time $(k, k+1, \dots, \tau)$, and dispatch solution at $k+\tau$ is related to uncertain demand at $(k+1, \dots, \tau)$, similar to the multi-stage robust optimization problem in~\cite{multistage}. However, the control laws considered in~\cite{multistage} are polynomial in past-observed uncertainties; in this work, we do not restrict the decision variables to be any forms of previous-observed uncertain demands. The dispatch decisions are numerical optimal solution of a robust optimization problem. With a list of parameters and variables shown in Table~\ref{T1_parameter}, considering both the current and future dispatch costs when making the current decisions, we define a robust taxi dispatch problem as the following
%\footnotesize
\begin{align}
\begin{split}
\underset{X^1}{\text{min}}\ &\underset{r^1\in \Delta_1}{\text{max}}\ \underset{X^2, L^2}{\text{min}}\ \underset{r^2\in\Delta_2}{\text{max}}\dots\underset{X^{\tau},L^{\tau}}{\text{min}}\ \underset{r^{\tau}\in\Delta_{\tau}}{\text{max}}\\
 \quad J %=&\sum_{k=1}^{\tau} (J_D(X^k)+\beta J_E(X^k,r^k))\\
=& \sum_{k=1}^{\tau} \sum_{i=1}^n\left( \sum\limits_{j=1}^n X^k_{ij} W_{ij}+\frac{\beta r^k_i}{\left (\sum\limits_{j=1}^{n} X^k_{ji}-\sum\limits_{j=1}^{n} X^k_{ij} + L^k_i \right)^{\alpha}}\right) \\
\text{s.t.}\quad
& (L^{k+1})^T=( \mathbf{1}^T_n X^k -(X^k \mathbf{1}_n)^T+(L^k)^T )P^k,\\%\quad k=1,\dots,\tau-1\\
&\mathbf{1}^T_n X^k -(X^k \mathbf{1}_n)^T+(L^k)^T \geqslant \mathbf{1}_n^T,\\
&  X^k_{ij} W_{ij} \leq m X^k_{ij}, \\%\forall i, j \in \{1, \dots, n\},\\
&X^k_{ij} \geq 0,\quad i,  j \in \{1, 2, \dots, n \}.
\end{split}
\label{Tm}
\end{align}
After getting an optimal solution $(X^{1})^*$ of~\eqref{Tm}, we adjust the solution by rounding methods to get an integer number of taxis to be dispatched towards corresponding regions. It does not affect the optimality of the result much in practice, since the objective or cost function is related to the demand-supply ratio of each region. A feasible integer solution of~\eqref{Tm} always exists, since $X^k_{ij}=0,\ \forall i,j,k$ is feasible. Although we cannot provide any theoretical guarantee on the suboptimality of the rounded integer solution, in the numerical experiments the costs under integer solution after rounding and the original real value optimal solution are comparable.

\section{Algorithm For Constructing Uncertain Demand Sets}
\label{sec:algorithm1}
With many factors affecting taxi demand during different time within different areas of a city, explicitly describing the model is a strict requirement and errors of the model will affect the performance of dispatch frameworks. Considering future demand and demand uncertainties benefits for minimizing worst-case demand-supply ratio mismatch error and idle distance~\cite{taxi_Feiiccps15, Miao_tase16}. It is then essential to construct a model that captures the spatial-temporal demand uncertainties and provides a probabilistic guarantee about the vehicle resource allocation cost. We construct demand uncertainty sets via Algorithm~\ref{alg: algorithm_1}---getting a sample set of $r_c$ from the original dataset and partition the sample set, bootstrapping a threshold for the test statistics according to the requirement of the probability guarantee, and calculating the model of uncertainty sets based on the thresholds. % In this section, we explain each step, summarize the process in .% and discuss factors to consider for choosing parameters of the algorithm. Numerical examples are shown in Section~\ref{sec:simulation}. However, only using a standard deviation range of demand~\cite{taxi_Feiiccps15, Miao_tase16} cannot tell how possible the true real-world cost is higher than the optimal cost. Hence, 

\subsection{An uncertainty set with probabilistic guarantee}
For convenience, we concisely denote all the variables of the taxi dispatch problem as $x$. Assume that we do not have knowledge about the true distribution $\mathbb{P}^*(r_c)$ of the random demand vector $r_c$. WIth the objective function $J(r_c, x)$ of problem~\eqref{Tm}, the probabilistic guarantee for the event that the true dispatch cost being smaller than the optimal dispatch cost is defined as the following chance constrained problem
 %We convert the objective function~\eqref{obj} to the following form
\begin{align}
\begin{split}
\underset{x}{\text{min}} \quad &M\\
\text{s.t.}\quad & P_{r_c \sim \mathbb{P}^*(r_c)} (f(r_c,x) = J(r_c,x)-M \leqslant 0)\geqslant 1-\epsilon.
\end{split}
\label{obj_2}
\end{align}
The constraint $f$ and objective function $J$ are concave in $r_c$ for any $x$, and convex in $x$ for any $r_c$. Without loss of generality about the objective and constraint functions, equivalently we aim to find solutions for%of the following form of chance constrained problem
\begin{align}
\begin{split}
\underset{x}{\text{min}}\ \quad &J(r_c, x)\\
\text{s.t.}\quad & P_{r_c \sim \mathbb{P}^*(r_c)}(f(r_c,x) \leqslant 0) \geqslant 1-\epsilon.
\end{split}
\label{chance}
\end{align}
When it is difficult to explicitly estimate $\mathbb{P}^*(r_c)$, we solve the following robust problem such that its optimal solutions satisfy the probabilistic guarantee requirement for~\eqref{chance}%given constraints $f(r_c,x)$ that are concave in $r_c$ for any $x$,
\begin{align}
\begin{split}
\underset{x}{\text{min}}\ \underset{r_c \sim \Delta}{\text{max}} \ \quad J (r_c, x),\quad 
\text{s.t.}\quad f(r_c, x) \leqslant 0.%\forall \quad r_c \in \mathcal{U},
\end{split}
\label{constraint}
\end{align}
%where $x\in \mathbb{R}^n$ is the optimization variable. Another requirement is that the robust optimization problem is computationally tractable problem with this uncertainty set. Hence,By solving~\eqref{constraint}, the performance of optimal solutions is guaranteed for $r_c\sim \mathbb{P}^*$.
Then $r_c$ of problem~\eqref{constraint} can be any vector in the uncertainty set $\Delta$ instead of a random vector in~\eqref{chance}. The uncertainty set that keeps the optimal solution of~\eqref{constraint} satisfying the constraints of problem~\eqref{chance} is defined as the following:
\begin{problem}
Construct an uncertainty set $\Delta, r_c\in \Delta$, given $0<\epsilon<1$ and samples of random vectors $r_c$, such that

(P1). The robust constraint~\eqref{constraint} is computationally tractable.

(P2). The set $\Delta$ implies a probabilistic guarantee for the true distribution $\mathbb{P}^*({r}_c)$ of a random vector ${r}_c$ at level $\epsilon$, that is, for any optimal solution $x^* \in \mathbb{R}^k$ and for any function $f(r_c,x)$ concave in $r_c$, we have the implication:
\begin{align}
\begin{split}
&\text{If}\ f(r_c,x^*)  \leq 0,\quad \text{for}\ \forall r_c\in\Delta,\\ 
 &\text{then}\ \mathbb{P}_{r_c \sim \mathbb{P}^*(r_c)}^*(f({r}_c, x^*)\leqslant 0) \geq 1-\epsilon. 
\end{split}
\label{epsilon}
\end{align} 
\label{prob_1}
\end{problem}
\vspace{-8pt}
%In this work, the objective of an uncertainty set is to guarantee the performance of robust taxi dispatch approaches, given 
The given probabilistic guarantee level $\epsilon$ is related to the degree of conservativeness of the robust optimization problem. %The trade-off between the average cost of robust optimal solutions and the probabilistic level evaluated in Section~\ref{sec:simulation}. 

\subsection{Aggregating demand and partition the sample set}
Every $\tau$ discretized time slots of demand $(r^t, \dots, r^{t+\tau})$ are concatenated to a vector $r_c (t)$. The first step is to transform the original taxi operational data to a dataset of sampled vector $\tilde{r}_c(d,t)$ of different dates $d$ for each index $t$. For instance, assume we choose the length of each time slot as one hour, and the dataset records all trip information of taxis during each day. According to the start time and GPS coordinate of each pick-up event, we aggregate the total number of pick up events during one hour at each region to get samples $\tilde{r}_c(d,t)$.

It is always possible to describe the support of the distribution of all samples contained in the dataset even they do not follow the same distribution, as explained in Figure~\ref{partition}. When there is prior knowledge or categorical information to partitioned the dataset into several subsets, we get a more accurate uncertainty set for each sub-dataset to provide the same probabilistic guarantee level compared with the uncertainty set from the entire dataset. Clustering algorithms with categorical information~\cite{Huang_kmodes} is applicable for dataset partition when information besides pick up events is available, such as weekdays/weekends, weather or traffic conditions.  It is worth noting that if the uncertainty sets are built for a categorical information set $\mathcal{I}=\{\mathcal{I}_1, \mathcal{I}_2,\dots\}$, then for the robust dispatch problems, we require the same set of categories is available in real-time, hence we apply the uncertainty set of $\mathcal{I}_1$ to find solutions when the current situation is considered as $\mathcal{I}_1$. 

 \begin{figure}[!t]
\centering
\includegraphics [width=0.30\textwidth]{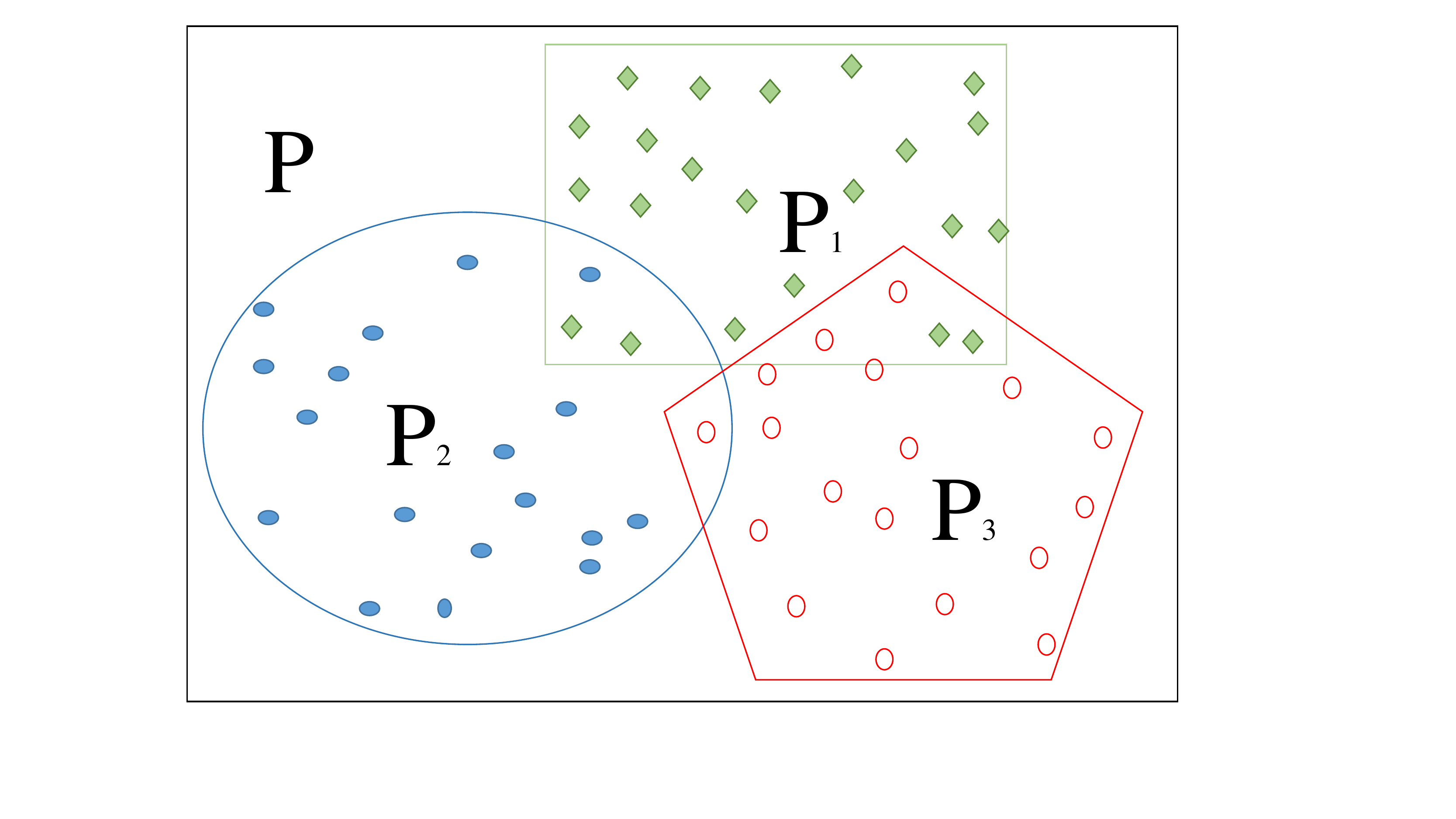}
\vspace{-15pt}
\caption{Intuition for partitioning the whole dataset. When the data set includes data from three distributions $P_1, P_2, P_3$, without prior knowledge, we can build a larger uncertainty set that describes the range of all samples in the dataset. The problem is that the uncertainty set is not accurate enough.} 
\label{partition}
\vspace{-15pt}
\end{figure}

\subsection{Uncertainty Modeling}
\label{sec: theory}
% the following problem is to decide a valid threshold that the hypothesis testing is accepted with repect to the data set and the probabilistic guarantee level. 
The basic idea to define an uncertainty set is to find a threshold for a hypothesis testing that is acceptable with respect to the given dataset and a required probabilistic guarantee level, and the formula of an uncertainty set is related to the threshold value of an acceptable hypothesis testing. Given the original data, the null hypothesis $H_0$, $\alpha_h$, and the test statistics $T$, we need to find a threshold that accepts $H_0$ at significance value $\alpha_h$ for each subset of sampled demand vectors. Since we do not assume that the marginal distribution for every element of vector $r_c$ is independent with each other, we apply two models without any assumptions about the true distribution $\mathbb{P}^*(r_c)$  in the robust optimization literature~\cite{datad_robust}~\cite{N_1970}~\cite{SC_2003} on the spatial-temporally correlated demand data. 

\subsubsection{Box type of uncertainty demand sets built from marginal samples}
\label{box_uncertain}
One intuitive description about a random vector is to define a range for each element. For instance, consider the following multivariate hypothesis holds simultaneously for $i=1,2,\dots, \tau n$ with given thresholds $\bar{q}_{i,0}, \underline{q}_{i,0} \in \mathbb{R}$~\cite{N_1970} 
\begin{align}
\begin{split}
H_{0,i} : &{inf\{t: \mathbb{P}(r_{c,i} \leqslant t)\geqslant 1-\frac{\epsilon}{\tau n}\}} \geqslant \bar{q}_{i,0}\\
              &{inf\{t: \mathbb{P}(-r_{c,i} \leqslant t)\geqslant 1-\frac{\epsilon}{\tau n}\}} \geqslant -\underline{q}_{i,0}.
%VaR_{\frac{\epsilon}{\tau n}}^{\mathbb{P}^*} (e_i) \geqslant \bar{q}_{i,0}\ \text{and}\  VaR_{\frac{\epsilon}{\tau n}}^{\mathbb{P}^*} (-e_i) \geqslant \underline{q}_{i,0}. %\quad i=1,\dots, 24n.
%of a random vector $r_c \in \mathbb{R}^{\tau n}$, 
\end{split}
\label{H_box}
\end{align}

% Let the estimated increasingly ordered value of $\hat{r}_{c,i}$ equal to , i.e., $\hat{r}^{(1)}_{c,i}=r^{(1)}_{c,i}, \dots, \hat{r}^{(N_B)}_{c,i}$.

Assume that we have $N_B$ random samples for each component $r_{c,i}$ of $r_c$, ordered in increasing value as $r^{(1)}_{c,i}, r^{(2)}_{c,i}, \dots, r^{(N_B)}_{c,i}$ no matter what is the original sampling order. We define the index $s$ by
\footnotesize
\begin{align}
s=min\left\{k \in \mathbb{N}: \sum\limits_{j=k}^{N_B}\left(\begin{array}{c}N_B\\j\end{array}\right) \left(\frac{\epsilon}{\tau n}\right)^{N_B-j} \left(1-\frac{\epsilon}{\tau n}\right)^j \leqslant \frac{\alpha_h}{2\tau n}\right\},
\label{s}
\end{align}
\normalsize
and let $s=N_B+1$ if the corresponding set is empty. 
The test $H_0$ is rejected if $r^{(s)}_{c,i}  \geqslant \bar{q}_{i,0}\  \text{or}\ - r_{c, i}^{(N_B-s+1)} \geqslant -\underline{q}_{i,0}$. To construct an uncertainty set, we need an accepted hypothesis test. Hence, we set $\bar{q}_{i,0}=r^{(s)}_{c,i}$ and $\underline{q}_{i,0}=r_{c, i}^{(N_B-s+1)}$. The following uncertainty set is then applied in this work based on the range hypothesis testing~\eqref{H_box}.
%worst case bound over the confidence region for each marginal $\mathcal{P}^M (r_c)$
%With $H_{0,i}$ related to the range of the $i$-th component of $r_c$,  then $H_{0,i}$~\eqref{H_box} is always accepted
%\newtheorem{theorem}{Theorem}
\begin{proposition}[\cite{datad_robust},~\cite{N_1970}]
If $s$ defined by equation~\eqref{s} satisfies that $N_B-s+1 < s$, then, with probability at least $1-\alpha_h$ over the sample, the set
\begin{align}
\mathcal{U}^M_{\epsilon} (r_c)=\left\{r_c \geqslant \mathbf{0}| r_{c,i}^{(N_B-s+1)}\leqslant r_{c,i} \leqslant r_{c,i}^{(s)}\right\}
\label{UM_box}
\end{align}
implies a probabilistic guarantee for $\mathbb{P}^*(r_c)$ at level $\epsilon$. 
%Moreover, the support function of $\mathcal{U}_{\epsilon}^M (r_c)$ is
%\begin{align*}
%\delta^*(v|\mathcal{U}_{\epsilon}^M (r_c)) =\sum\limits_{i=1}^{\tau n} max(v_i\hat{r}_{c,i}^{(N-s+1)}, v_i\hat{r}_{c,i}^{(s)}).
%\end{align*}
\label{theorem_7}
\end{proposition}

\subsubsection{SOC type of uncertainty set motivated by moment hypothesis testing}%~\cite{SC_2003}}
\label{soc_uncertain}
%Though the box type of uncertainty set reflects the spatial-temporal correlations by varying range values with different dimensions of $r_c$ (value of $\tau n$), 
It is not easy to tell directly from the uncertainty set~\eqref{UM_box} when the range of one component changes how will others be affected. To directly show the spatial-temporal correlations of the demand, we also apply hypothesis testing related to both the first and second moments of the true distribution $\mathbb{P}^*(r_c)$ of the random vector~\cite{SC_2003}. 
%\\\centerline{$
\begin{align}
H_0: \mathbb{E}^{\mathbb{P}^*} [r_c] = r_0 \ \text{and}\ \mathbb{E}^{\mathbb{P}^*} [r_c r_c^T] -\mathbb{E}^{\mathbb{P}^*} [r_c]\mathbb{E}^{\mathbb{P}^*} [r_c^T] =\Sigma_0,
%$}
\label{H_0_soc}
\end{align}
%with test statistics $T$ defined as $\|\hat{r}_c-r_0\|$ and $\|\hat{\Sigma}-\Sigma_0\|$. 
where $r_0$ and $\Sigma_0$ are the (unknown) true mean and covariance of $r_c$,  $\mathbb{E}^{\mathbb{P}^*} [r_c]$ and $\mathbb{E}^{\mathbb{P}^*} [r_c r_c^T]$ are the estimated mean and covariance from data. Without knowledge of $r_0$ and $\Sigma_0$, $H_0$ is rejected when the difference among the estimation of mean or covariance according to multiple times of samples is greater than the threshold, i.e., $\|\mathbb{E}^{\mathbb{P}}[\tilde{r}_c]-\hat{r}_c\|_2 > \Gamma_1^B$ or $\|\mathbb{E}^{\mathbb{P}}[\tilde{r}_c\tilde{r}_c^T] -\mathbb{E}^{\mathbb{P}}[\tilde{r}_c]\mathbb{E}^{\mathbb{P}}[\tilde{r}_c^T]-\hat{\Sigma}\|_F > \Gamma_2^B$, where $\mathbb{E}^{\mathbb{P}}[\tilde{r}]$ is the estimated mean value of one experiment, $\hat{r}_c$ and $\hat{\Sigma}$ are the estimated mean and covariance from multiple experiments, $\Gamma_1^B$ and $\Gamma_2^B$ are the thresholds. The remaining problem is then to find the values of the thresholds such that hypothesis testing~\eqref{H_0_soc} holds given the dataset. The uncertainty set derived based on the moment hypothesis testing is defined in the following proposition.
%In the following Section~\ref{algorithm_sec}, the detailed steps of calculating the thresholds $\Gamma_1^B$ and $\Gamma_2^B$ at a desired significance value $\alpha_h$ and probabilistic guarantee level $\epsilon$ based on the given dataset is described.
%\footnote{Bootstrapped thresholds and theoretic bounds proposed by work~\cite{LR_2010} are compared in~\cite{datad_robust}. The bootstrapped thresholds result in a smaller uncertainty set in general, hence reduces the ambiguity in $\mathbb{P}^*$. In this work, we apply the bootstrapped thresholds $\Gamma_1^B$ and $\Gamma_2^B$ based on the dataset.}.
\begin{proposition}[\cite{datad_robust},~\cite{SC_2003}]
With probability at least $1-\alpha_h$ with respect to the sampling, the following uncertainty set $\mathcal{U}_{\epsilon}^{CS}(r_c)$ implies a probabilistic guarantee level of $\epsilon$ for $\mathbb{P}^*(r_c)$ 
\begin{align}
\begin{split}
\mathcal{U}_{\epsilon}^{CS} (r_c)=&\{r_c\geqslant \mathbf{0}| r_c= \hat{r}_c + y + C^T w: \exists y, w \in \mathbb{R}^{n\tau} \ s.t.\\
&\quad \|y\|_2 \leqslant \Gamma_1^B, \|w\|_2 \leqslant\sqrt{\frac{1-\epsilon}{\epsilon}}\},
\end{split}
\label{u_cs}
\end{align}
where $C^TC=\hat{\Sigma} + \Gamma_2^B \mathbf{I}$ is a Cholesky decomposition. 
\label{theorem_10}
\end{proposition} 
When one component of $r_c$ increases or decreases, we have an intuition how it affects the value of other components of $r_c$ by the expression~\eqref{u_cs}.
%By testing the properties of both first and second moments of the dataset, the uncertainty set~\eqref{u_cs} reflects the spatial-temporal correlations of the demand model directly compared with the box type~\eqref{UM_box}. 

\subsection{Algorithm}
\label{algorithm_sec}
With a threshold of the test statistics calculated via the given dataset, we then apply the formula~\eqref{UM_box} for constructing a box type of uncertainty set, and the formula~\eqref{u_cs} for an SOC type of uncertainty set, respectively. The following Algorithm~\ref{alg: algorithm_1} describes the complete process for constructing uncertain demand sets based on the original dataset.

\begin{algorithm}
\caption{Constructing uncertain demand sets}
\textbf{Input: A dataset of taxi operational records}\\
\textbf{1. Demand aggregating and sample set partition}\\
Aggregate demand to get a sample set $\mathcal{S}$ of the random demand vector $r_c$ from the original dataset. Partition the sample set $\mathcal{S}$ and denote a subset $\mathcal{S}(t, I_p)\subset \mathcal{S}$, $p=1,\dots,P$ as the subset partitioned  for each time index $t$ according to either prior knowledge or categorical information $I_p$. %Denote the partitioned sample subset for each time index $t$ as $\mathcal{S}(t,I_p)$. %process the original data and get a subset of sampled vectors $\tilde{r}_c(d_l, t, I_p)$. 
\\\textbf{2. Bootstrapping thresholds for test statistics}
\\\textbf{for} each subset $\mathcal{S}(t, I_p)$\ \textbf{do}
\\ \textbf{Initialization:} { Testing statistics $T$, a null-hypothesis $H_0$, the probabilistic guarantee level $\epsilon$, a significance level $0 < \alpha_h <1$, the number of bootstrap time $N_B \in \mathbb{Z}_+$.}%A subset of random sampled demand vectors for every $\tau$ time slots $\mathcal{S}(t,I_p)=\{\tilde{r}_c (d_1, t,I_p),\tilde{r}_c(d_2, t,I_p), \dots\}$,
%\\\textbf{Output:} {} 
%\\ \quad \textbf{Initialization}
\\\textbf{Estimate the mean $\hat{r}_c(t,I_p)$ and covariance $\hat{\Sigma}(t,I_p)$ for vector $r_c$ based on subset $\mathcal{S}(t, I_p)$.} 
\\ \quad \textbf{for}\ $j=1,\dots, N_B$\ \textbf{do}
\\ \quad\quad (1). Re-sample $\mathcal{S}^j(t,I_p)=\{\tilde{r}_c (d_1, t,I_p),\dots,\tilde{r}_c(d_{N_B}, t,I_p)\}$ data points from $\mathcal{S}(t,I_p)$ with replacement for each $t$.
\\ \quad\quad (2). Get the value of the test statistics based on $\mathcal{S}^j(t,I_p)$. %$T^j \leftarrow T(\mathcal{S}^j, H_0)$.
%$\Gamma_1(j,t)=\sum\limits_{l=1}^N\left\|(r_c(d_l,t)-\frac{1}{N}\sum\limits_{l=1}^N r_c\right\|_2$, and $\Gamma_2(j,t)=\|\sum\limits_{l=1}^N (r_c(d_l, t)r^T_c(d_l, t)- )-\hat{\Sigma}\|_2$
\\ \quad \textbf{end for}
\\ \quad (3). Get the thresholds of the $\alpha$ significance level for $H_0$.
\\ \quad \textbf{end for}
%get $\lceil N_B(1-\alpha)\rceil$- largest value of $T^1,\dots, T^{N_B}$:
\\\textbf{3. Calculate the model of uncertainty sets}\\
Get the box type and the SOC type of uncertainty sets according to~\eqref{UM_box} and~\eqref{u_cs}, respectively, for each $t$ and $I_p$. 
% According to the bootstrapped thresholds and the equations of  chosen type of uncertainty mode, get the closed form equations for every $\tau$ time slots of demand vector $r_c$ of every partitioned subset of data.}
\textbf{Output: Uncertainty sets for problem~\eqref{Tm}}
\label{alg: algorithm_1}
\end{algorithm}

 %to estimated the true mean and covariance for calculating $\Gamma_1^B$ and $\Gamma_2^B$ in the following repeated sampling process
We do not restrict the method of estimating  mean $\hat{r}_c(t,I_p)$ and covariance $\hat{\Sigma}(t,I_p)$ matrices of a subset $\mathcal{S}(t, I_p)$ in step $2$,  and bootstrap is one method. For step $2.(2)$, the process for the box type of uncertainty sets is: calculate index $s$ that satisfies~\eqref{s} with the given $\epsilon$, sort each component of sampled vectors $r_c(d_l, t, I_p)$, and get the order statistics $r_{c,i}^{(N_B-s+1)}(j,t, I_p)$, $r_{c,i}^{(s)}(j,t,I_p)$ of the $j$-th sample set $\mathcal{S}^j(t, I_p)$. For the SOC type, we calculate the mean and covariance of the samples of the vector according to the subset $\mathcal{S}^j(t,I_p)$ as $\hat{r}_c (j,t,I_p)$ and $\hat{\Sigma}(j,t,I_p)$, respectively.

In step $2.(3)$, the $\alpha_h$ level thresholds for the box type of uncertainty sets are the $\lceil N_B(1-\alpha_h)\rceil$-th largest value of the upper bound $r_{c,i}^{(s)}(j,t,I_p)$ and the $\lceil N_B \alpha_h \rceil$-th largest value of the lower bound $r_{c,i}^{(N_B-s+1)}(j,t,I_p)$ for the $i$-th component. For the SOC type of uncertainty sets, we calculate the mean and covariance of $r_c(t, I_p)$ for the $N_B$ times bootstrap as $\hat{r}_c(t,I_p)$ and $\hat{\Sigma}(t,I_p)$, and get $\Gamma_1(j,t,I_p)=\|\hat{r}_c(j,t,I_p)-\hat{r}_c(t,I_p)\|_2$, $\Gamma_2(j,t,I_p)=\|\hat{\Sigma}(j,t,I_p)-\hat{\Sigma}(t,I_p)\|_2.$ Denote the $\lceil N_B(1-\alpha_h)\rceil$-th largest value of $\Gamma_1(j,t,I_p)$ and $\Gamma_2(j,t,I_p)$as $\Gamma_1^B (t,I_p)$ and $\Gamma_2^B(t,I_p)$, respectively. 

In summary, to construct a spatial-temporal uncertain demand model for problem~\eqref{Tm}, in this section, we consider the taxi operational record of each day as one independent and identically distributed (i.i.d.) sample for the concatenated demand vector $r_c$. By partitioning the entire dataset to several subsets according to categorical information such as weekdays and weekends, we are able to build uncertainty sets for each subset of data without additional assumptions about the true distribution of the spatial-temporal demand profile. Then we design Algorithm~\ref{alg: algorithm_1} to construct a box type and an SOC type of uncertainty sets based on data that provide a desired probabilistic guarantee of robust solutions.

\section{Computationally Tractable Formulations}
\label{sec:algorithm}
We build equivalent computationally tractable formulations of problem~\eqref{Tm} with different definitions of uncertain sets calculated by Algorithm~\ref{alg: algorithm_1} in this section. Hence, the robust taxi dispatch problem considered in this work can be solved efficiently. Computational tractability of a robust linear programming problem for ellipsoid uncertainty sets is discussed in~\cite{robustconvex}. The process is to reformulate constraints of the original problem to its equivalent convex constraints that must hold given the uncertainty set. The objective function of problem~\eqref{Tm} is concave of the uncertain parameters $r^k$, convex of the decision variables $X^k, L^k$ with the decision variables on the denominators, not standard forms of linear programming (LP) or semi-definite programming (SDP) problems that already covered by previous work~\cite{robustconvex, datad_robust}. Hence, we prove one equivalent computationally tractable form of problem~\eqref{Tm} for each uncertainty set constructed in Section~\ref{sec:algorithm1}.

%%%%%%%%%%%%%%%%%%%%%%%%%
%%%%%%%%%%%%%%%%%%%%%%%%%%
%%%%%%%%%%%%%%%%%%%%%
%%%%%%%%%%%%%%%%%%%%%%%%%which is a general ~\eqref{u_cs}uncertain parameter sets regarding to this specific formulations. 

Only the $J_E$ components of objective functions in~\eqref{Tm} include uncertain parameters, and the decision variables of the function are in the denominator of the function $J_E$. %We ~\eqref{UM_box}~\eqref{u_cs}
The box type uncertainty set defined as~\eqref{UM_box} is a special form of polytope, hence, we first prove an equivalent standard form of convex optimization problem for~\eqref{Tm} for a polytope uncertainty set as the following.
%We derive the for one-stage robust optimization~\eqref{Tm} in the following theorem, and the explicit form of multi-stage problems with time independent uncertainty sets are obtained based on Theorem~\ref{T1_convex}. 
\begin{theorem} (Next step dispatch)
If the uncertainty set of problem~\eqref{Tm} when $\tau=1$ is defined as the non-empty polytope $\Delta : =\{ r \geq 0, Ar \leq b\}$, and we omit the superscripts $k$ for variables and parameters without confusion. Then problem~\eqref{Tm} with $\tau=1$ is equivalent to the following convex optimization problem
 \begin{align}
\begin{split}
\underset{X\geq 0,\lambda \geq 0}{\text{min}} \quad &\sum_{i=1}^{n} \sum_{j=1}^{n}X_{ij} W_{ij}+b^T\lambda \\
\text{s.t.}\quad\quad &A^T\lambda-\beta\begin{bmatrix} \frac{1}{\left (\sum\limits_{j=1}^{n} X_{j1}-\sum\limits_{j=1}^{n} X_{1j} + L_1 \right)^\alpha} \\ \vdots \\ \frac{1}{\left (\sum\limits_{j=1}^{n} X_{jn}-\sum\limits_{j=1}^{n} X_{nj} + L_n \right)^\alpha}\end{bmatrix}    \geq 0,\\
&\mathbf{1}^T_n X-X \mathbf{1}_n +L^T \geqslant 1,\\
%& \sum_{j} X_{ij} \leq L_i, i=1,\dots, n, \\
&  X_{ij} W_{ij} \leq m X_{ij},\\
& X_{ij}\geq 0,\quad \forall i, j \in \{1, \dots, n\}.
%&X_{ij} \in \mathbb{N}, i,  j \in \{1, 2, \dots, n \}.
\end{split}
\label{conv_T1}
\end{align} 
 \label{T1_convex}
\end{theorem}
\begin{proof}
See Appendix~\ref{appendix_T1}.
\end{proof}

To directly use the demand uncertainty set that describes the spatial-temporal correlation of $(r^1, \dots, r^{\tau})$ like~\eqref{UM_box} and~\eqref{u_cs} for the concatenated demand $r_c$ in problem~\eqref{Tm}, we first consider to group the maximization over each $r^k$ together to save the process of projection $r_c \in \Delta$ for individual $r^k \in \Delta_k$. Furthermore, we can find the dual (a minimizing problem) of the maximizing cost problem over $r_c \in \Delta$, and then numerically efficiently solve~\eqref{Tm} that minimizes the total cost during time $(1,2,\dots, \tau)$ under uncertain demand $r_c$. Hence, we first prove that the minimax equality holds for the maximin problem over each pair of $k$ and $k+1$ for problem~\eqref{Tm}, and~\eqref{Tm} is equivalent to the robust optimization problem shown in the following lemma. % The computationally tractable convex optimization forms of problem~\eqref{Tm} are then proved based on this theorem. 
\begin{lemma}(Minimax equality)
Given the assumption that the definition of the uncertainty sets $r_c \in \Delta$ and $r^k\in \Delta_k$ are compact (closed and convex), the robust dispatch problem~\eqref{Tm} is equivalent to the following robust dispatch problem
\begin{align}
\begin{split}
\underset{X^{1:\tau}, L^{2:\tau}}{\text{min}}\ \underset{r_c\in \Delta}{\text{max}}\quad &J =\sum_{k=1}^{\tau} (J_D(X^k)+\beta J_E(X^k,r^k))\\
%=& \sum_{k=1}^{\tau} \sum_{i}\left( \sum_{j} X^k_{ij} W_{ij}+\frac{\beta r^k_i}{(\mathbf{1}_n^T X^k_{\cdot i}-X^k_{i\cdot}\mathbf{1}_n+ L^k_i)^{\alpha}}\right) \\
\text{s.t.}\quad &\text{constraints of~\eqref{Tm}}, k=1,\dots, \tau.
%& (L^{k+1})^T=( \mathbf{1}^T_n X^k -(X^k \mathbf{1}_n)^T+(L^k)^T )P^k,\\%\quad k=1,\dots,\tau-1\\
%&\mathbf{1}^T_n X^k -(X^k \mathbf{1}_n)^T+(L^k)^T >0,\\
%&  X^k_{ij} W_{ij} \leq m X^k_{ij}, \\%\forall i, j \in \{1, \dots, n\},\\
%&X^k_{ij} \geq 0,\quad i,  j \in \{1, 2, \dots, n \}.
\end{split}
\label{Tm_minimax}
\end{align}
%Here $L^1$ is the initial number of vacant taxis within each region before dispatch provided by sensor information, not a decision variable, and we omit the time index of $L^k, k=2,\dots, \tau$ in minimization for notation convenience.
\label{lemma_minimax}
\end{lemma}
\begin{proof}
See Appendix~\ref{appendix_minimax}.
\end{proof}

For the robust optimization problem~\eqref{Tm}, the computationally tractable convex form depends on the definition of uncertainty sets.When conditions of Lemma~\ref{lemma_minimax} hold, equivalent convex optimization forms of problem~\eqref{Tm} are derived based on problem~\eqref{Tm_minimax}. For a multi-stage robust optimization problem that restricts the near-optimal control input of linear dynamical systems to be a certain degree of polynomial of previous observed uncertainties, an approximated semidefinite programming method for calculating the time dependent control input is proposed in~\cite{multistage}. The method does not require minimax equality holds for the robust optimal control problem.

The box type uncertainty set~\eqref{UM_box} is a special form of polytope, that the uncertain demand model during different time of a day is described separately. The process of converting problem~\eqref{Tm} to an equivalent computationally tractable convex form is similar to that of the one-stage robust optimization problem. The result is described as the following lemma.
\begin{lemma}
If the uncertain set for $r^k, k=1,\dots, \tau$ describes each demand vector $r^k$ separately as a non-empty polytope with the form
\begin{align}
\Delta_k: =\{ r^k \geq 0, A_k r^k \leq b_k\},\quad k=1,\dots, \tau,
\label{polytope}
\end{align} 
%and we relax the integer constraint for variable $X^k$ to $X^k_{ij} \geq 0$,
problem~\eqref{Tm} is equivalent to the following convex optimization problem
\begin{align}
\begin{split}
\underset{X^k,\lambda^k, L^k\geq 0}{\text{min}} \quad& \sum_{k=1}^{\tau}(\sum_{i}^n \sum_{j}^n X^k_{ij} W_{ij}+b_k^T\lambda^k )\\
\text{s.t.}\quad\quad %& L^k=p^k(X^{k-1}, L^{k-1}), k=2,\dots T,\\
&A_k^T\lambda^k-\beta\begin{bmatrix} \frac{1}{\left (\sum\limits_{j=1}^{n} X^k_{j1}-\sum\limits_{j=1}^{n} X^k_{1j} + L^k_1 \right)^{\alpha}} \\ \vdots \\ \frac{1}{\left (\sum\limits_{j=1}^{n} X^k_{jn}-\sum\limits_{j=1}^{n} X^k_{nj} + L^k_n \right)^{\alpha}}\end{bmatrix}    \geq 0, \\
&\text{constraints of~\eqref{Tm}}, k=1,\dots, \tau.
\end{split}
\label{conv_Tm}
\end{align}
\label{lemma_Tm_poly}
\end{lemma}
\begin{proof}
See Appendix~\ref{appendix_Tm_poly1}.
\end{proof}

For a more general case that the uncertainty sets for $r^1, \dots, r^{\tau}$ are temporally correlated, the following theorem and proof describe the equivalent computationally tractable convex form of~\eqref{Tm}. 
\begin{theorem}
When $\Delta$ is defined as the following non-empty polytope set
\begin{align}
\Delta : =\{(\Delta_1, \dots,\Delta_{\tau}) | A_1 r^1+\dots + A_{\tau} r^{\tau} \leq b, r^k \geq 0\},
\label{delta}
\end{align}
problem~\eqref{Tm} is equivalent to the following convex optimization problem 
\begin{align}
\begin{split}
\underset{X^k, L^k, \lambda \geq 0}{\text{min}} \quad&\sum_{k=1}^{\tau}(\sum_{i}^n \sum_{j}^n X^k_{ij} W_{ij})+b^T\lambda \\
\text{s.t.}\quad &A_k^T\lambda-\beta\begin{bmatrix} \frac{1}{\left (\sum\limits_{j=1}^{n} X^k_{j1}-\sum\limits_{j=1}^{n} X^k_{1j} + L^k_1 \right)^{\alpha}} \\ \vdots \\ \frac{1}{\left (\sum\limits_{j=1}^{n} X^k_{jn}-\sum\limits_{j=1}^{n} X^k_{nj} + L^k_n \right)^{\alpha}}\end{bmatrix}    \geq 0,\\
&\text{constraints of~\eqref{Tm}},\quad k=1,\dots, \tau.
%%%%%%%%%%%%%%%
%%%%%%%%%%%%%
%&L^k=p^k(X^{k-1}, L^{k-1}), k=2,\dots T,\\
%&\mathbf{1}^T_n X^k +(L^k)^T >0,\\
%& \sum_{j} X^k_{ij} \leq L^k_i, i=1,\dots, n, \\
%& X^k_{ij} W_{ij} \leq \alpha X^k_{ij}, \forall i, j \in \{1, \dots, n\}.
%%%%%%%%%%%%%
%%%%%%%%%%%%%%%%
\end{split}
\label{conv_Tm_dep}
\end{align}
\label{Tm_poly}
\end{theorem}
\begin{proof}
See Appendix~\ref{appendix_Tm_poly2}.
\end{proof}

With an uncertain demand model defined as~\eqref{u_cs} for concatenated $r^1, \dots, r^{\tau}$, the following theorem derive the equivalent computationally tractable form of problem~\eqref{Tm}.
\begin{theorem}
When the uncertainty set for $r^1, \dots, r^{\tau}$ is defined as the SOC form of~\eqref{u_cs}, problem~\eqref{Tm} is equivalent to the following convex optimization problem~\eqref{u_cs_opt}.
\begin{align}
\begin{split}
\underset{X^k, L^k, z}{\text{min}} \quad&\sum_{k=1}^{\tau}\sum_{i}^n \sum_{j}^n X^k_{ij} W_{ij}\\%+\beta t \\
                                                                  &+\beta \left(\hat{r}^T_c z+\Gamma_1^B \|z\|_2 +\sqrt{\frac{1}{\epsilon}-1} \|Cz\|_2\right)\\                                       
\text{s.t.}\quad  & c_l(X) \leqslant z,\\                                 
%&\|u\|_2 \leqslant \lambda, \lambda \geqslant 0,\\
&\text{constraints of~\eqref{Tm}},\quad k=1,\dots, \tau,
\end{split}
\label{u_cs_opt}
\end{align}
where $c_l(X) \in \mathbb{R}^{\tau n}$ is the concatenation of $c(X^1), \dots, c(X^{\tau})$.
\label{theorem_soc}
\end{theorem}
\begin{proof}
See Appendix~\ref{appendix_soc}.
\end{proof}
%%%%%%%%%
%%%%%%%%%

It is worth noting that any optimal solution for problem~\eqref{nonr} has a special form between any pair of regions $(i, q)$.  
\begin{proposition}
Assume $X^{1*}, \dots, X^{\tau*}$ is an optimal solution of~\eqref{nonr}, then any $X^{k*}$ satisfies that for any pair of $(p,q)$, at least one value of the two elements $X^{k*}_{qi}$ and $X^{k*}_{iq}$ is $0$.
\end{proposition}
\begin{proof}
We prove by contradiction. Assume that one optimal solution has the form $X^{k}$ such that $X^{k}_{qi} >0$ and $X^{k}_{iq} >0$. Without loss of generality, we assume that $X^k_{qi} \geq X^k_{iq}$, and let 
\begin{align*}
X^{k*}_{qi} = X^k_{qi}-X^k_{iq},  X^{k*}_{iq}=0,
\end{align*} 
other elements of $X^{k*}$  equal to $X^k$. Then
\begin{align*}
 &\sum_{j=1}^n X^k_{ji}-\sum_{j=1}^n X^k_{ij} 
= X^k_{qi}-X^k_{iq}+\sum_{j \neq q} X^k_{ji} -\sum_{j \neq q} X^k_{ij}\\
= & X^{k*}_{qi} + 0+\sum_{j \neq q} X^{k*}_{ji} -\sum_{j \neq q} X^{k*}_{ij}
= \sum_{j=1}^n X^{k*}_{ji}-\sum_{j=1}^n X^{k*}_{ij},\\
&\sum_{j=1}^n X^k_{ji}-\sum_{j=1}^n X^k_{ij}+ L^k_i=\sum_{j=1}^n X^{k*}_{j i}-\sum_{j=1}^n X^{k*}_{ij}+ L^k_i.
\end{align*}                                                                                 
Hence, we have
%\begin{align*}
 $J_E(X^k, r^k)= J_E(X^{k*}, r^k). $
%\end{align*}
All constraints are satisfied and $X^{k*}$ is also a feasible solution for~\eqref{Tm}.

Next, we compare $J_D(X^k)$ and $ J_D(X^{k*})$. With $X^{k}_{qi}>X^{k}_{iq}>0$, and
$X^{k*}_{qi} = X^k_{qi}-X^k_{iq} \geq 0$, we have
\begin{align*}
X^k_{qi} > X^{k*}_{qi},\ %\ \text{and}\\
%\end{align*} 
%\begin{align*}
X^{k}_{qi}W_{qi} + X^k_{iq}W_{iq}> X^{k*}_{qi} W_{qi}+X^{k*}_{iq} W_{iq}.
\end{align*}
Thus the partial cost 
%\begin{align*}
$J_D(X^k) > J_D(X^{k*})$,
%\end{align*}
which contradicts with the assumption that $X^k$ is an optimal solution. To summarize, we show that an optimal solution cannot have $X^k_{qi}>0, X^k_{iq}>0$ at the same time, and at least one of  $X^{k*}_{qi}$ and $X^{k*}_{iq}$ should be $0$.
\end{proof}

With equivalent convex optimization forms under different uncertainty sets, robust taxi dispatch problem~\eqref{Tm} is computationally tractable and solved efficiently.

\section{Data-Driven Evaluations}
\label{sec:simulation}

We conduct data-driven evaluations based on four years of taxi trip data of New York City~\cite{Dan_nyc}. A summary of this data set is shown in Table~\ref{datanyc}. In this data set, every record represents an individual taxi trip, which includes the GPS coordinators of pick up and drop off locations, and the date and time (with precision of seconds) of pick-up and drop-off locations. The dispatch solutions based on different granularities of equal-area region partitions have been compared in~\cite{Miao_tase16}, and other region partition methods are discussed in~\cite{Miao2017}. In the following experiments, we use equal-area grid partition since it is a baseline, and compare the robust and non-robust solutions based on the same region partition method. One partition example given the map of Manhattan area is shown in Figure~\ref{nyc}, where we visualize the density of taxi passenger demand with the data we use for large-scale data-driven evaluations. The lighter the region, the higher the daily demand density, and the middle regions typically have higher density than the uptown and downtown regions. We construct uncertainty sets according to Algorithm~\ref{alg: algorithm_1}, discuss factors that affect modeling of the uncertainty set, and compare optimal costs of the robust dispatch formulation~\eqref{Tm} and the non-robust optimization form~\eqref{nonr} in this section.   

\textbf{How vacant taxis are balanced across regions with different $\alpha$ values}: Figure~\ref{alpha_r} shows mismatch between supply and demand defined as~\eqref{mismatch} for different optimal solutions of minimizing $J_E$ defined in~\eqref{JE} for $\alpha \in (0,1]$. With $\alpha$ closer to $0$, the optimal value of~\eqref{mismatch} is smaller. We choose $\alpha=0.1$ for calculating optimal solutions of~\eqref{Tm} and~\eqref{nonr} in this section.
%Lemma~\ref{alpha_balance} proves the conclusion that the value of $\alpha$ affects the optimal solution, and we can select the value of $\alpha$ according to how we want to adjust supply to meet demand. Figure~\ref{alpha_r} shows mismatch between supply and demand defined as equation~\eqref{mismatch} for $\alpha \in (0,1]$. With $\alpha$ closer to $0$, the value of~\eqref{mismatch} is smaller, which is consistent with Lemma~\ref{alpha_balance}. We choose $\alpha=0.1$ for calculating optimal solutions of~\eqref{Tm} and~\eqref{nonr} in this section.

%\begin{figure}[!t]
\begin{figure}[b!]
\vspace{-8pt}
\centering
\includegraphics [width=0.3\textwidth]{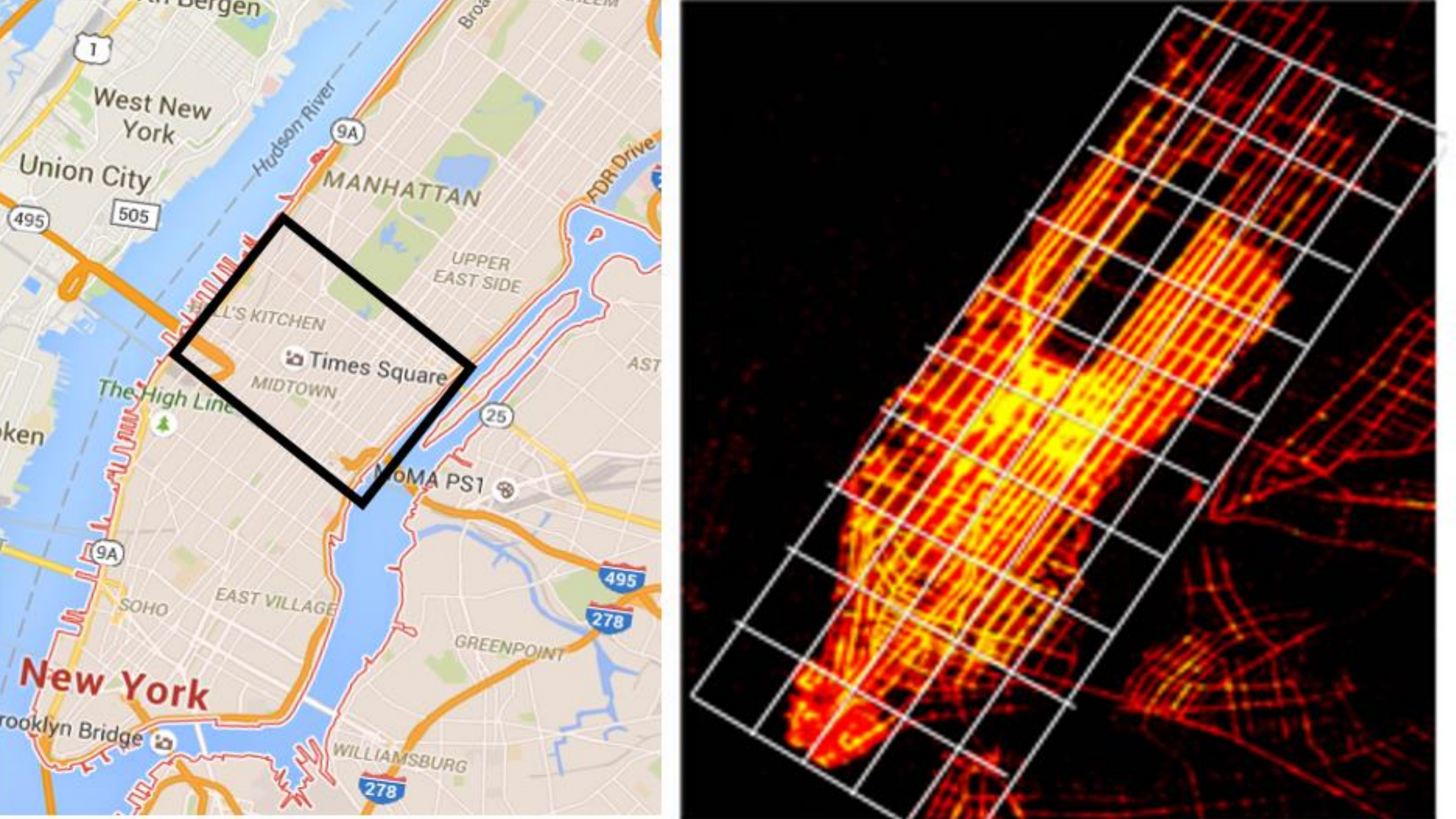}
\vspace{-10pt}
\caption{Map of Manhattan area in New York City.}% One example of a region partition method is a rectangular region: the X-axis and Y-axis for region partition is skewed according to the street direction.} 
\label{nyc}
\end{figure}
\begin{figure}[b!]
\vspace{-10pt}
\centering
\includegraphics [width=0.33\textwidth]{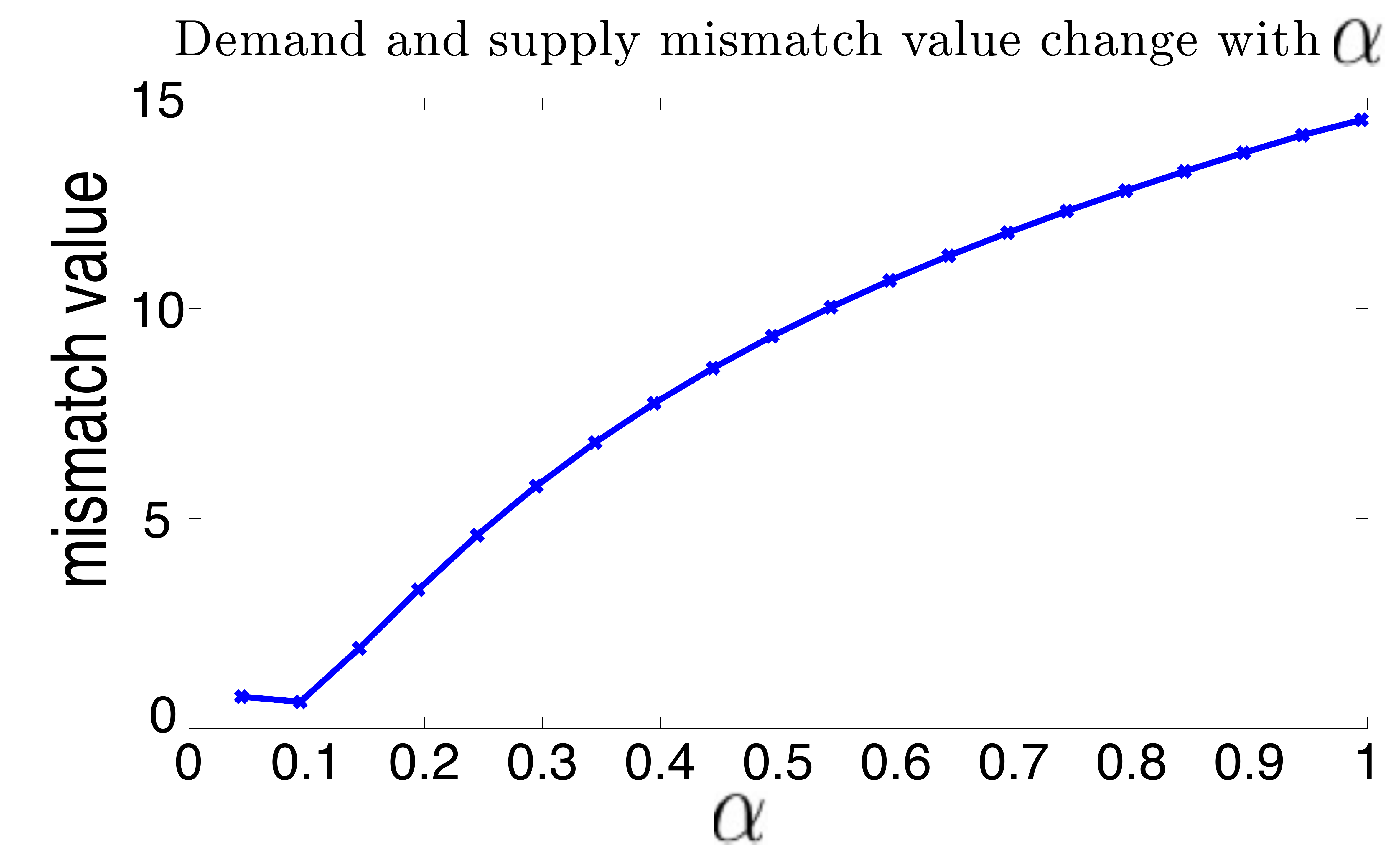}
\vspace{-15pt}
\caption{Comparison of demand and supply mismatch values defined as~\eqref{mismatch} with different solutions for minimizing $J_E$ defined in~\eqref{JE} with $\alpha$ in range $(0,\ 1]$. The value of function~\eqref{mismatch} under an optimal solution of $J_E$ is smaller with an $\alpha$ closer to $0$, which means the dispatch solution tends to be more balanced throughout the entire city.}
% distribution of vacant taxis is not balanced without dispatch method, since the supply/demand ratio among regions vary much.} 
\label{alpha_r}
%\vspace{-10pt}
\end{figure} 
\begin{table*}[t!]
%\vspace{-8pt}
\centering
\begin{tabular}{|c|c|c|c|c|c|}
  \hline
 \multicolumn{3}{|c|} {Taxi Trip Data set}                         & \multicolumn{3}{|c|}{Format} \\ \hline
  Collection Period            & Data Size & Record Number           & ID & Trip Time            &Trip Location  \\ \hline
 01/01/2010-12/31/2013 & $100GB$ & about $7$ million         & Date &Start and end time  & GPS coordinates of start and end\\ \hline
 \end{tabular}
 \caption{New York city data in the evaluation section.}
\label{datanyc}
\vspace{-15pt}
\end{table*}

%and we consider records of each date as one sample.
\begin{figure}[!t]
%\begin{figure}[b!]
%\vspace{-8pt}
\centering
\includegraphics [width=0.35\textwidth]{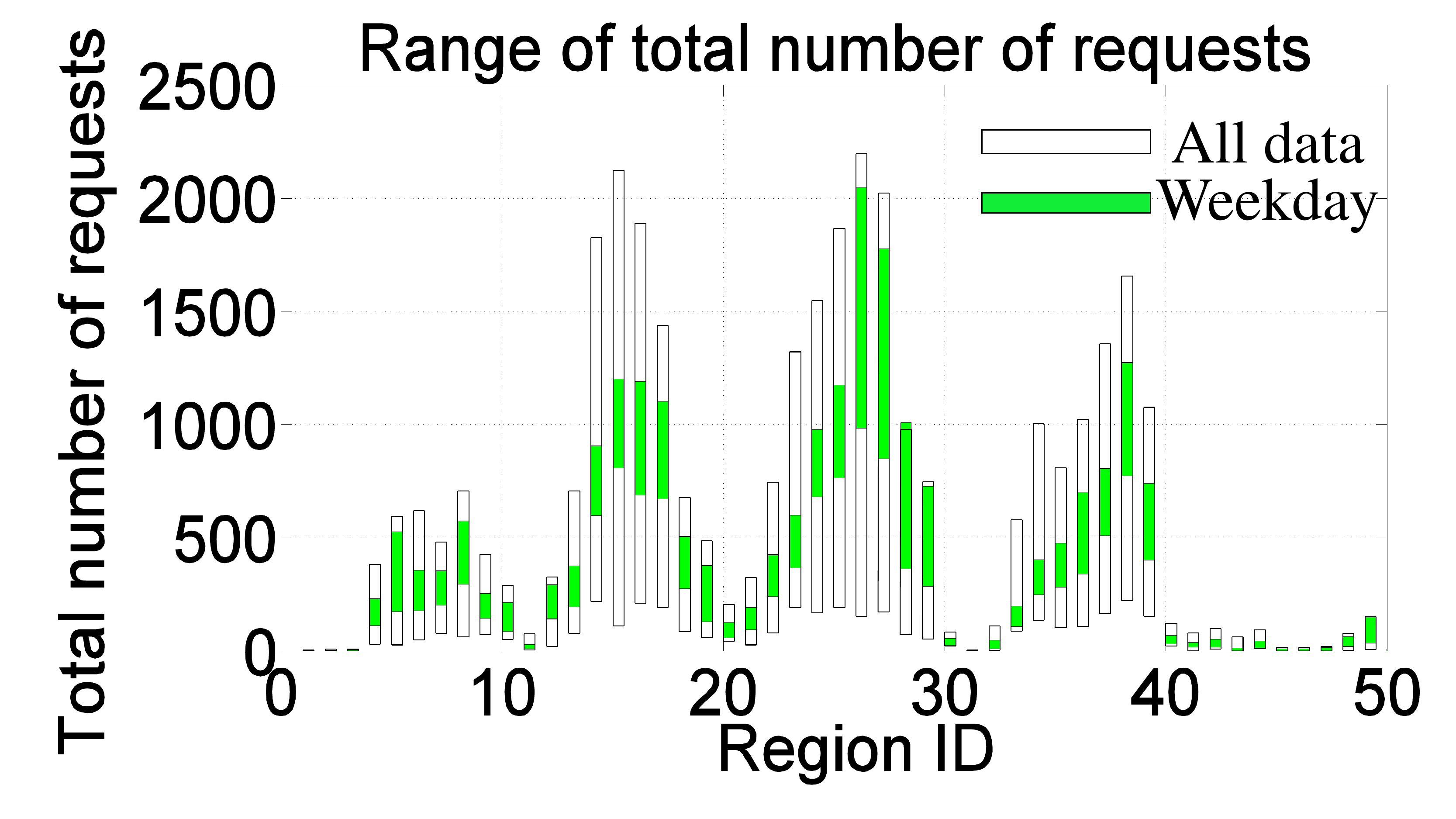}
\vspace{-10pt}
\caption{Comparison of box type of uncertainty sets constructed from all data and those constructed only based on trip records of weekdays. When keeping all parameters the same, by applying data of weekdays, the range of uncertainty set for each $r_{c,i}$ is smaller than that based on the whole dataset.} 
\label{box_weekday}
\vspace{-8pt}
\end{figure}
\begin{figure}[!t]
\centering
\includegraphics [width=0.35\textwidth]{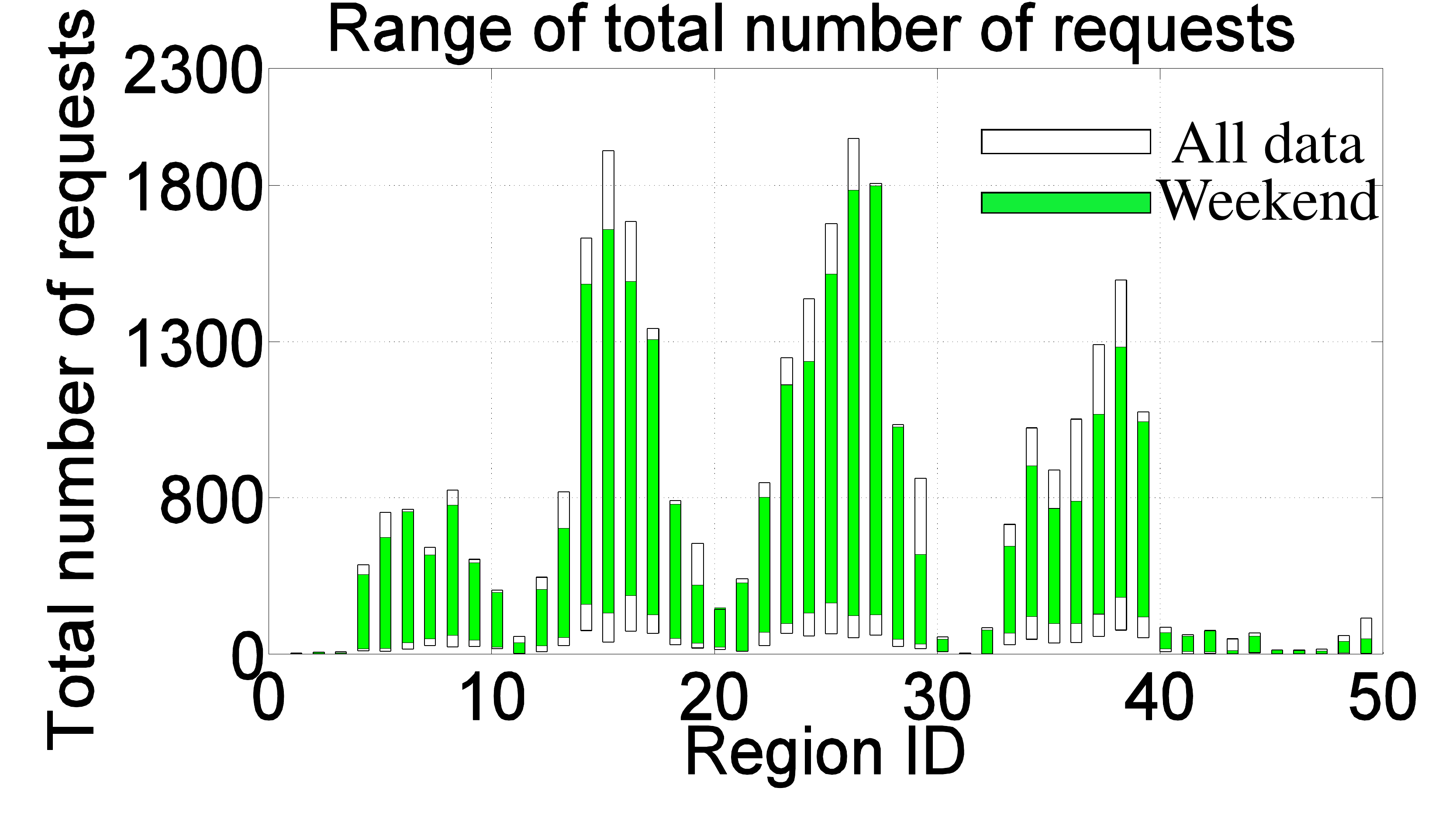}
\vspace{-10pt}
\caption{Comparison of box type of uncertainty sets constructed from all data and uncertainty sets constructed only based on trip records of weekends.}%The range of uncertainty set built from weekend records only (green bar) is smaller than the result based on the whole dataset.} 
\label{box_weekend}
\vspace{-10pt}
\end{figure}

\subsection{Box type of uncertainty set}
For all box type of uncertainty sets shown in this subsection with the model described in Subsection~\ref{box_uncertain}, we set the confidence level of hypothesis testings as $\alpha_h=10\%$, bootstrap time as $N_b=1000$, number of randomly sampled data (with replacement) for each time of bootstrap as $N_B=10000$.%, and .% for uncertainty sets In Figure~\ref{box_weekday} and Figure~\ref{box_weekend},$\mathcal{U}_{\epsilon}$. %be the probabilistic guarantee level of 
%, i.e., $\hat{r}_{c,i}^{(N-s+1)}$ and $\hat{r}_{c,i}^{(s)}$ in the box type of uncertainty set

\textbf{Partitioned dataset compared with non-partitioned dataset}:
We show the effects of partitioning the trip record dataset by weekdays and weekends in Figure~\ref{box_weekday} and~\ref{box_weekend}. The whole city is partitioned into $50$ regions, the prediction time horizon is $\tau =4$, where one time instant means one hour, $\epsilon=0.3$, and every $r_c \in \mathbb{R}^{200 \times 1}$.  Figures~\ref{box_weekday} and~\ref{box_weekend} show the lower and upper bounds of each region during one time slot of~\eqref{UM_box}.  By applying data of weekdays and weekends separately, the range $[r_{c,i}^{(s)},\ r_{c,i}^{(N_B-s+1)}]$ of each component is reduced. To get a measurement of the uncertainty level, we defined the sum of range of every component for $r_c$  as\\
%\begin{align}
\centerline{$U (r_c)= \sum\limits_{i=1}^{\tau n}(r_{c,i}^{(s)}-r_{c,i}^{(N_B -s+1)})$.}
%\label{range_rc}
%\end{align}
For the box type of uncertainty sets, when values of the dimension of $r_c$, i.e., $\tau n$, $\alpha_h$ and $\epsilon$ are fixed, a smaller $U(r_c)$ means a smaller area of the uncertainty set, or a more accurate model. We denote $U(r_c)$ calculated via records of weekdays and weekends as $U_{wd}(r_c)$ and $U_{wn}(r_c)$ respectively, compared with $U(r_c)$ constructed from the complete dataset, we have $\frac{U(r_c)-U_{wd}({r}_c)}{U(\hat{r}_c)}=52\%$, $\frac{U({r}_c)-U_{wn}({r}_c)}{U(\hat{r}_c)}=28\%.$ This result shows that when by constructing an uncertainty set for each subset of partitioned data, we reduce the range of uncertainty sets to provide the same level of probabilistic guarantee for the robust dispatch problem. This is because samples contained in each subset of data do not follow the same distribution and can be categorized as two clusters. %If the uncertainty sets constructed from some  
%main reason that this partition works is 
%compared with constructing an uncertain set with non-partitioned data.

%This means, when keeping all the parameters the same, to obtain a certain level of probability guarantee $\epsilon$ for the robust dispatch problem, partitioning the dataset helps reduce the uncertainty of the demand model. 
%For the box type of uncertainty sets, each component of $r_c$ is described as a range of positive numbers. To get the same probability level guarantee $\epsilon$, the range of each component built partitioned dataset is smaller than that of the non-partitioned dataset. 

\textbf{Choose an appropriate $N_B$ for high-dimensional $r_c$:}
It is worth noting that the index $s$ affects the range selection for every component $r_{c,i}$, hence, for different values of $\alpha_h, \epsilon, \tau, n$, we should adjust the number of samples $N$ to get an accurate estimation of the marginal range. As shown in Table~\ref{table_nt},  $N$ need to be large enough for a large $\tau n$ value, or $s$ is too close to $N$ and the upper and lower bounds ${r}_{c,i}^{(N_B-s+1)}$, ${r}_{c,i}^{(s)}$ cover almost the whole range of samples. Hence, the box type uncertainty set is not a good choice for large $\tau n$ value, though the computational cost of solving problem~\eqref{conv_Tm_dep} is smaller than that of~\eqref{u_cs_opt} with the same size of $\tau n$.
\begin{table}[t!]
%\vspace{-8pt}
\centering
\begin{tabular}{|c|c|c|c|c|c|}
  \hline
   $N_B$              & $\alpha_h$     &$\epsilon$      & $n$            &$\tau$             & $s$ \\ \hline
$10000$          &$0.1$            &$0.2$             &  $50$        & $2$      & $9992$ \\ \hline
$10000$          &$0.1$            &$0.5$             &  $50$        & $2$      & $9970$   \\ \hline
$10000$          &$0.3$            &$0.2$             &  $50$        & $2$      & $9991$ \\ \hline
$10000$          &$0.1$            &$0.2$             &  $1000$    & $2$      & $9999$   \\ \hline
$10000$          &$0.1$            &$0.5$             &  $1000$        & $2$   & $9999$  \\ \hline
 \end{tabular}
     \caption{Value of index $s$ for the box type uncertainty set~\eqref{s}. For large $\tau n$, $N_B$ need to be large, or $s$ is too close to $N_B$ that the range covers values of almost all samples.}
     \label{table_nt}
    \vspace{-15pt}
\end{table}

\subsection{SOC type of uncertainty set}
The SOC type of uncertainty set is a high-dimensional convex set that is not able to be plotted. The bootstrapped thresholds for the hypothesis testing to construct the SOC uncertainty sets based on partitioned and non-partitioned data are summarized in Table~\ref{table_sepa}. Similarly as the box type of uncertainty sets, when we separate the dataset and construct an uncertainty demand model for weekdays and weekends respectively, the sets are smaller compared to the uncertain demand model for all dates. When $\alpha$ and $\epsilon$ values are fixed, with smaller $\Gamma_1^B$ and $\Gamma_2^B$, the demand model $\mathcal{U}_{\epsilon}^{CS}$ is more accurate to guarantee that with at least probability $1-\epsilon$, the constraints of the robust dispatch problems are satisfied. Numerical results of this conclusion are shown in Table~\ref{table_sepa}. 
%are $\Gamma_1^B=265.3$ and $\Gamma_2^B=3864.7$. 
 
\begin{table}[t!]
%\vspace{-8pt}
\centering
\begin{tabular}{|c|c|c|c|}
  \hline
Data type& Weekdays & Weekends &Non partitioned\\ \hline
%\multicolumn{2}{|c|} }   & \multicolumn{2}{|c|} {}       & \multicolumn{2}{|c|}{} \\ \hline
 $\Gamma_1^B$ &10.53        &13.84          & 17.96   \\ \hline
 $\Gamma_2^B$ & 2576.94   &2923.35  &3864.47    \\ \hline
% idle distance               &3.056 &1.718   &1.096  &4. 519  \\ \hline
%total cost                &0.645   & 5.434   &13.009    & 47.854  \\ \hline  
 \end{tabular}
     \caption{Comparing thresholds with and without discriminating weekdays and weekends data. When $\Gamma_1^B$ or $\Gamma_2^B$ is smaller, the volume of the uncertainty set is smaller. Here $n=1000$, $\tau =3$, $N_B=1000$, $\epsilon=0.3$, $\alpha_h=0.2$.}
     \label{table_sepa}
  \vspace{-20pt}
\end{table}
%for providing a certain level of probabilistic guarantee is more accurate. 

\textbf{How $n$ and $\tau$ affect the accuracy of uncertainty sets}:
For a box type of uncertainty set, when $\tau n$ is a large value, the bootstrap sample number $N_B$ should be large enough such that index $s$ is not too close to $N$. Without a large enough sample set, we choose to construct an SOC type of uncertainty set (such as $\tau n=1000,  N_B=10000$ in Table~\ref{table_nt}). Since SOC captures more information about the second moment properties of the random vector compared with the box type uncertainty set, some uncorrelated components of $r_c$ will be reflected by the estimated covariance matrix, and the volume of the uncertainty set will be reduced. We show the value of $\Gamma_1^B$ and $\Gamma_2^B$ with different dimensions of $r_c$ or $\tau n$ values in table~\ref{table_nt}. When increasing the value of $\tau n$, values of $\Gamma_1^B$ and $\Gamma_2^B$ are reduced, which means the uncertainty set is smaller. However, it is not helpful to reduce the granularity of region partition to a smaller than street level, since we construct the model for a robust dispatch framework and a too large $n$ is not computationally efficient for the dispatch algorithm.

\begin{table}[t!]
%\vspace{-8pt}
\centering
\begin{tabular}{|c|c|c|}
  \hline
                      &           $\Gamma_1^B$           & $\Gamma_2^B$ \\ \hline
%\multicolumn{2}{|c|} }   & \multicolumn{2}{|c|} {}       & \multicolumn{2}{|c|}{} \\ \hline
$n=50, \tau=1$                       &$42.37$         &$1.52 \times 10^5$ \\ \hline
$n=50, \tau=3$                       & $52.68$       & $4.29\times 10^4$   \\ \hline
$n=50, \tau=6$                       & $107.35$     &$8.23 \times 10^5$    \\ \hline
$n=10, \tau=3$                       & $71.35$      & $ 3.56\times 10^5$    \\ \hline
$n=1000, \tau=3$                   & $10.53$       &$2576.94$    \\ \hline
% idle distance               &3.056 &1.718   &1.096  &4. 519  \\ \hline
%total cost                &0.645   & 5.434   &13.009    & 47.854  \\ \hline  
 \end{tabular}
     \caption{Comparing thresholds of SOC uncertainty sets for different dimensions $r_c$, by changing either the region partition number $n$ or the prediction time horizon $\tau$.}%Here $N=1000$, $\epsilon=0.3$ and $\alpha=0.2$. When $\Gamma_1^B$ or $\Gamma_2^B$ is smaller, the uncertainty set for providing the required level $\epsilon$ guarantee of the robust dispatch solutions is more accurate.}
     \label{table_nt}
\vspace{-20pt}
\end{table}

%We do not know the true probability distribution $\mathbb{P}^*$ of the demand vector $r_c$,  hence, we do empirically test on the feasibility of constraint, by counting the number of sampled $\tilde{r}_c \in \mathcal{S}_r$ satisfies that $f(\tilde{r}_c, x) \leqslant 0 $. The probability level $\epsilon$ is guaranteed according to the samples.

\subsection{Compare robust solutions with non-robust solutions}
In the experiments,  the idle geographical distance of one taxi between a drop-off event and the following pick-up event is approximated as one norm distance between the $2$D geographical coordinates (provided as longitude and latitude values of GPS data in the dataset) of the two points. Then the corresponding idle miles on ground is converted from the geographical distance according to the geographical coordinates of New York City. To test the quality of the uncertainty sets applied in the robust dispatch problems, we use the idea of cross-validation from machine learning. The dataset is separated as a training set for building the uncertain demand model, and a testing set for comparing the results of the dispatch solutions. The customer demand models applied in the robust and non-robust optimization problems are different. For the non-robust dispatch problem, the demand prediction $r^k$ is a deterministic vector. For instance, in this work we use the average or mean of the bootstrapped value of the training dataset. The non-robust dispatch solution for each time $k$ is calculated by solving the convex optimization form of dispatch problem formulated in work~\cite{taxi_Feiiccps15, Miao_tase16} with deterministic demand model. For all the experiments, we let $\beta=10$, $\alpha=0.1$ in problem~\eqref{Tm} to calculate the optimal solutions.

In the robust dispatch problem, the penalty function directly includes the uncertain demand $r^k$ is  for violating a balanced demand-supply ratio requirement. For each testing data $r^k$, we denote the demand-supply ratio mismatch error of a dispatch solution as~\eqref{mismatch}. We then compare the value of~\eqref{mismatch} of robust dispatch solutions with the SOC type of uncertainty set constructed in this work with the value of~\eqref{mismatch} of non-robust solutions of testing samples. The distribution of values are shown in Figure~\ref{cost_rdratio}. The average demand-supply ratio error is reduced by $31.7\%$ with robust solutions. We compare the cost distribution of total idle distance in Figure~\ref{cost_total}. It shows the average total idle distance is reduced by $10.13\%$. For all testing, the robust dispatch solutions result in no idle distance greater than $0.8 \times 10^5$, and non-robust solutions has $48\%$ of samples with idle distance greater than $0.8 \times 10^5$.  The cost of robust dispatch~\eqref{Tm} is a weighted sum of both the demand-supply ratio error and estimated total idle driving distance, and the average cost is reduced by $11.8\%$ with robust solutions. It is worth noting that the cost is calculated based on the integer vehicle dispatch solution after rounding the real value optimal solution of~\eqref{Tm}, and the cost is only $1\%$ higher than the optimal cost of~\eqref{Tm}. The performance of the system is improved when the true demand deviates from the average historical value considering model uncertainty information in the robust dispatch process. It is worth noting that the number of total idle distance shown in this figure is the direct calculation result of the robust dispatch problem. When we convert the number to an estimated value of corresponding miles in one year, the result is a total reduction of $20$ million miles in NYC.

\begin{figure}[t!]
%\vspace{-8pt}
\centering
\includegraphics [width=0.36\textwidth]{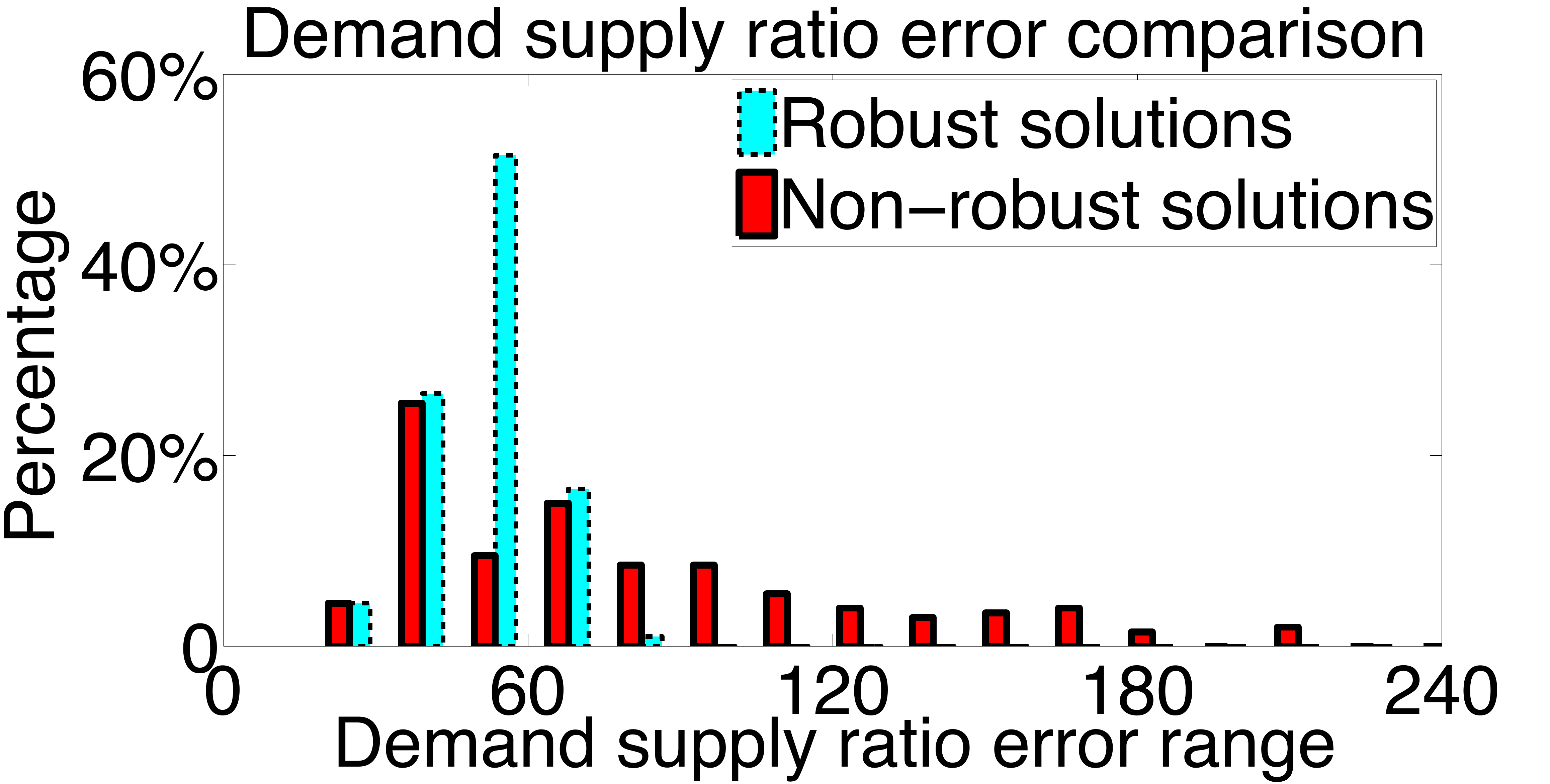}
\vspace{-10pt}
\caption{Demand-supply ratio error distribution of the robust optimization solutions with the SOC type of uncertain demand set ($\epsilon=0.25$, or probabilistic guarantee level $75\%$) and non-robust optimization solutions. The demand-supply ratio error of robust solutions is smaller than that of the non-robust solutions, that  the average demand-supply ratio error is reduced by $31.7\%$.} %the error greater than $60$ is reduced $45\%$, and $[25,\ 50],\  (40,55], \dots, (220,235]$ of two methods.  The bars show the percentage of experiments with a demand supply ratio error falling in each interval.
%(120-30/200)
%Robust optimization solutions in this work has a shorter tail than non-robust solutions.}
\label{cost_rdratio}
\vspace{-10pt}
\end{figure} 
\begin{figure}[t!]
%\vspace{-8pt}
\centering
\includegraphics [width=0.36\textwidth]{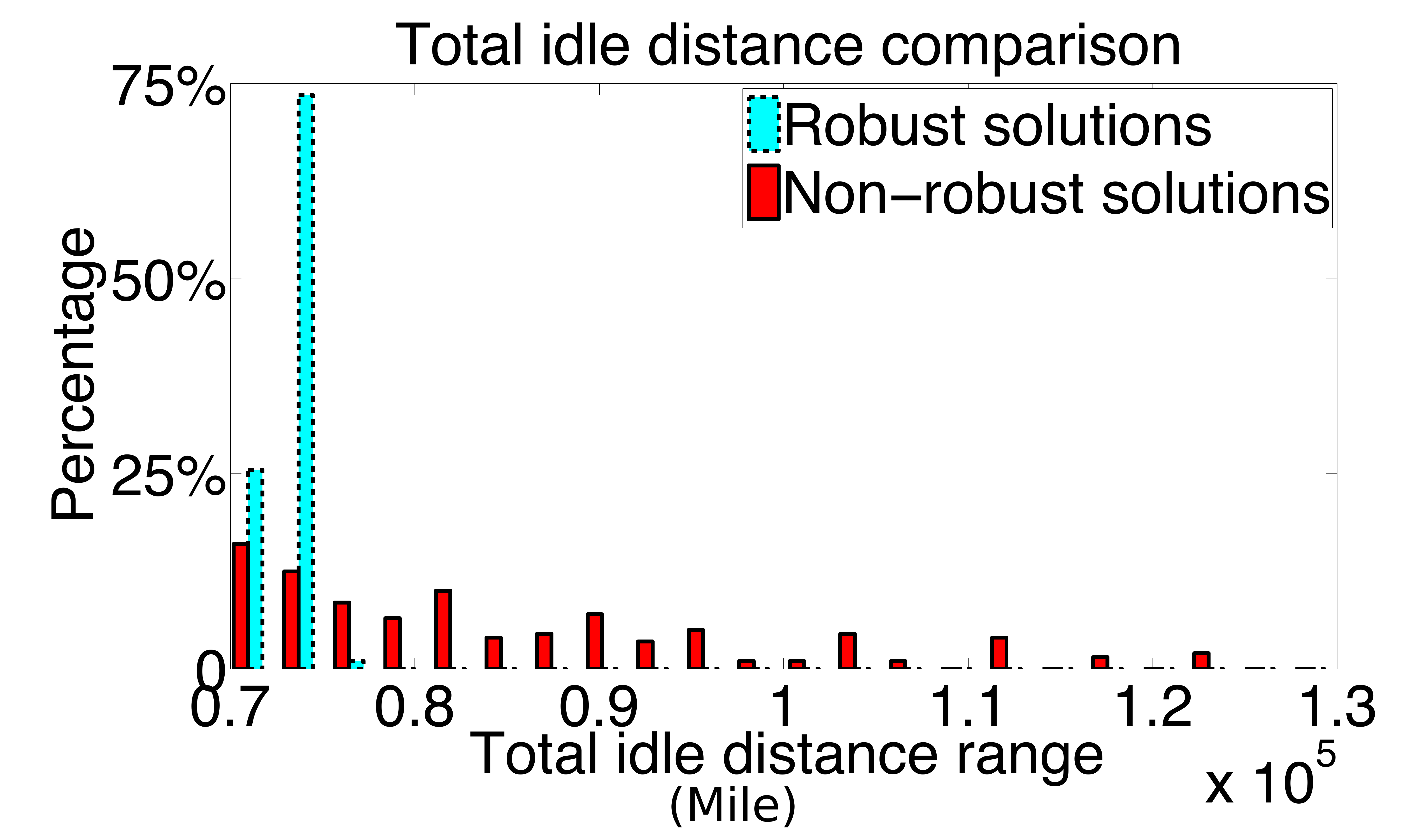}
\vspace{-10pt}
\caption{Total idle distance comparison of robust optimization solutions with the SOC type of uncertain demand set ($\epsilon=0.25$, or probabilistic guarantee level $75\%$) and non-robust optimization solutions. The average total idle distance is reduced by $10.13\%$. For all samples used in testing, the robust dispatch solutions result in no idle distance greater than $0.8 \times 10^5$, and non-robust solutions has $48\%$ of samples with idle distance greater than $0.8 \times 10^5$. The number of total idle distance shown in this figure is the direct calculation result of the robust dispatch problem, and we convert the number to an estimated value of corresponding miles in one year, the result is a total reduction of $20$ million miles in NYC.}%The average total cost is reduced by $11.8\%$ with robust solutions, and%$[0.7\times 10^5,\ 0.75\times 10^5],\  (0.75 \times 10^5, 0.8\times 10^5], \dots, (1.2 \tims 10^5, 1.25 \times 10^5]$ 
%The bars show the percentage of $200$ times of experiments with cost falling in each interval.
%a shorter tail than non-robust solutions, with  experiments of cost greater than $$.}
\label{cost_total}
\vspace{-8pt}
\end{figure} 

\textbf{Check whether the probabilistic level $\epsilon$ is guaranteed:}
%It is required that for any random vector $r_c$ that satisfies the true distribution $\mathbb{P}^*$ consistent with the sample dataset.  For the two types of uncertainty sets constructed in this work, we test whether the constraints hold by using samples that were not included for constructing the set, with the idea of cross-validation.  
Theoretically, the optimal solution of the robust dispatch problems with the uncertainty set should guarantee that with at least the probability $(1-\epsilon)$, when the system applies the robust dispatch solutions, the actual dispatch cost under a true demand is smaller than the optimal cost of the robust dispatch problem. Figures~\ref{epsilon_guarantee} and~\ref{epsilon_guarantee_soc} show the cross-validation testing result that the probabilistic guarantee level is reached for both box type and SOC type of uncertainty sets via solving~\eqref{conv_Tm_dep} and~\eqref{u_cs_opt}, respectively. Comparing these two figures, one key insight is that the robust dispatch solution with an SOC type uncertainty set provides a tighter bound on the probabilistic guarantee level that can be reached under the true random demand compared with solutions of the box type uncertainty set. It shows the advantage of considering second order moment information of the random vector, though the computational cost is higher to solve problem~\eqref{u_cs_opt} than to solve problem~\eqref{conv_Tm_dep}.
\begin{figure}[!t]
\centering
\includegraphics [width=0.38\textwidth]{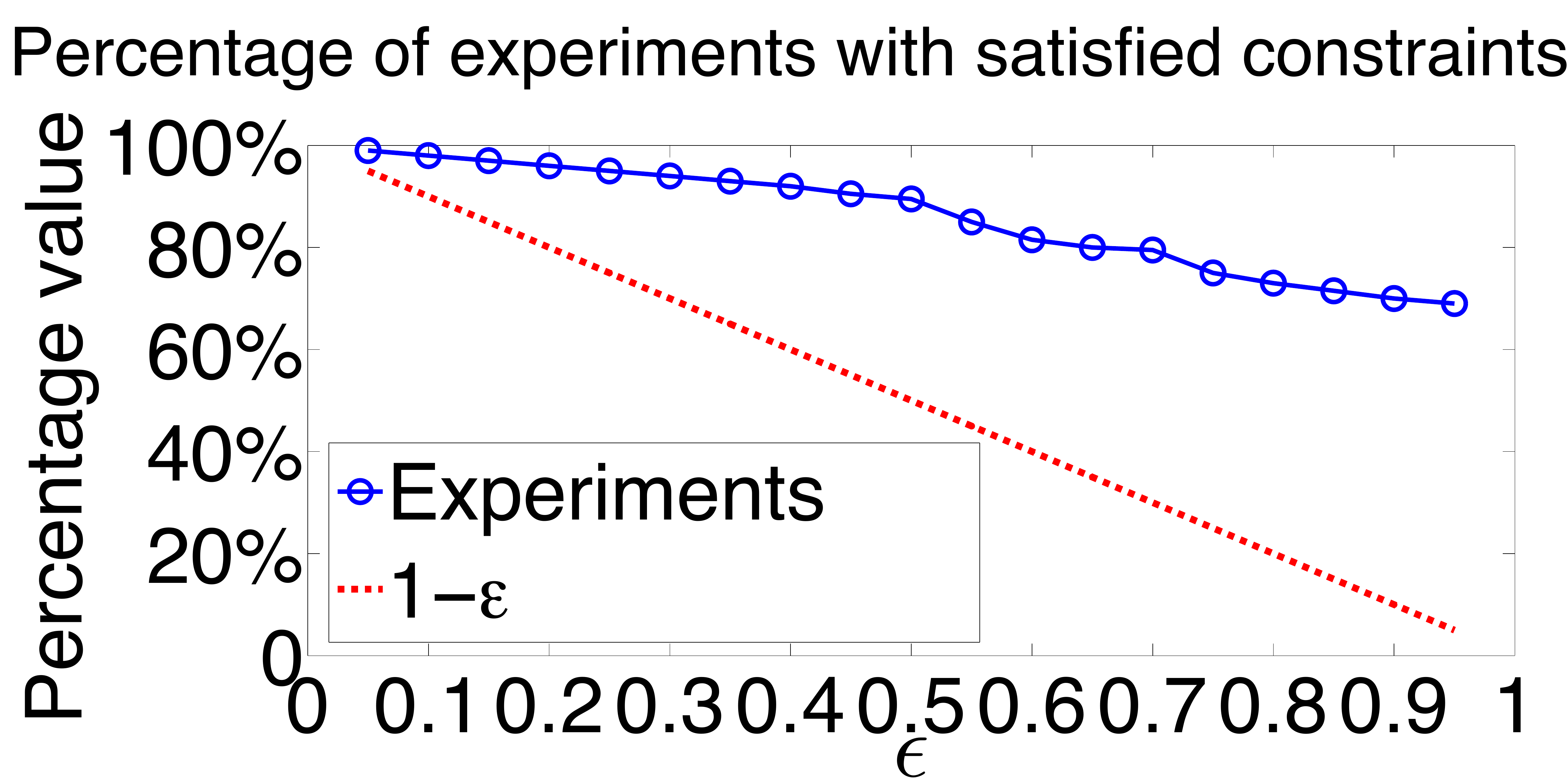}
\vspace{-5pt}
\caption{The percentage of tests that have a smaller true dispatch cost than the optimal cost of the robust dispatch problem with the box type uncertainty set constructed from data. When $1-\epsilon$ decreases, the percentage value also decreases, but always greater than $1-\epsilon$.} 
\label{epsilon_guarantee}
%\vspace{-8pt}
\end{figure}
\begin{figure}[!t]
\centering
\includegraphics [width=0.38\textwidth]{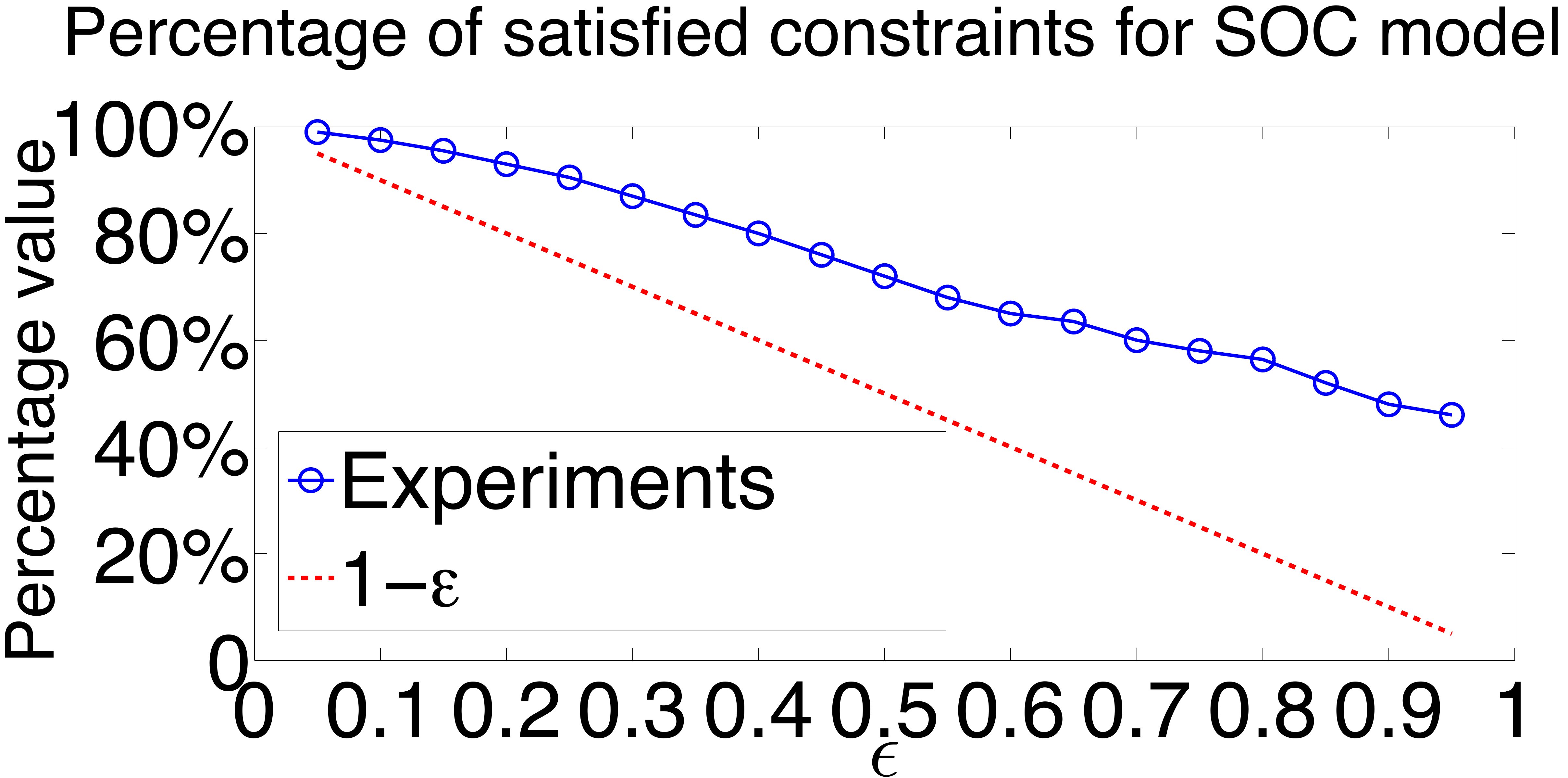}
\vspace{-5pt}
\caption{The percentage of tests that have a smaller true dispatch cost than the optimal cost of the robust dispatch problem with the SOC type of uncertainty set. When $1-\epsilon$ decreases, the percentage value also decreases, but always greater than $1-\epsilon$. The true percentage value is closer to the value of $1-\epsilon$ compared with the solution given a box type uncertainty set.} 
\label{epsilon_guarantee_soc}
\vspace{-10pt}
\end{figure}
\begin{figure}[!t]
\centering
\includegraphics [width=0.40\textwidth]{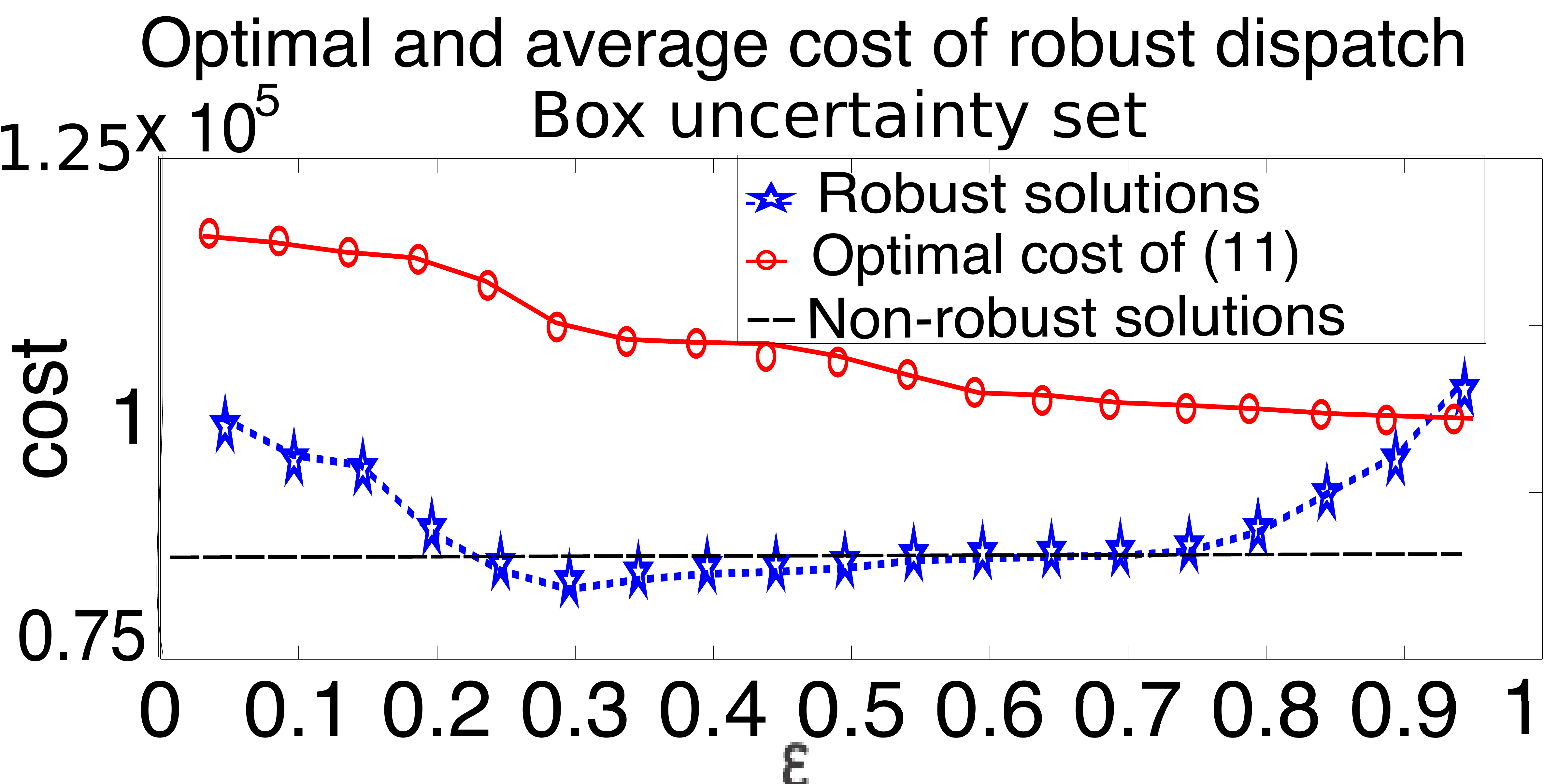}
\vspace{-5pt}
\caption{Comparison of the optimal cost of the robust dispatch problem with box type of uncertainty set and the average cost when applying the robust solutions for the test subset of sampled $r_c$. When $\epsilon=0.3$ the average cost is the smallest.}% The uncertainty set is the box type. There exists trade-off between probabilistic guarantee and the average cost of the robust solutions.} 
\label{epsilon_cost}
\vspace{-5pt}
\end{figure}
\begin{figure}[!t]
\centering
\includegraphics [width=0.40\textwidth]{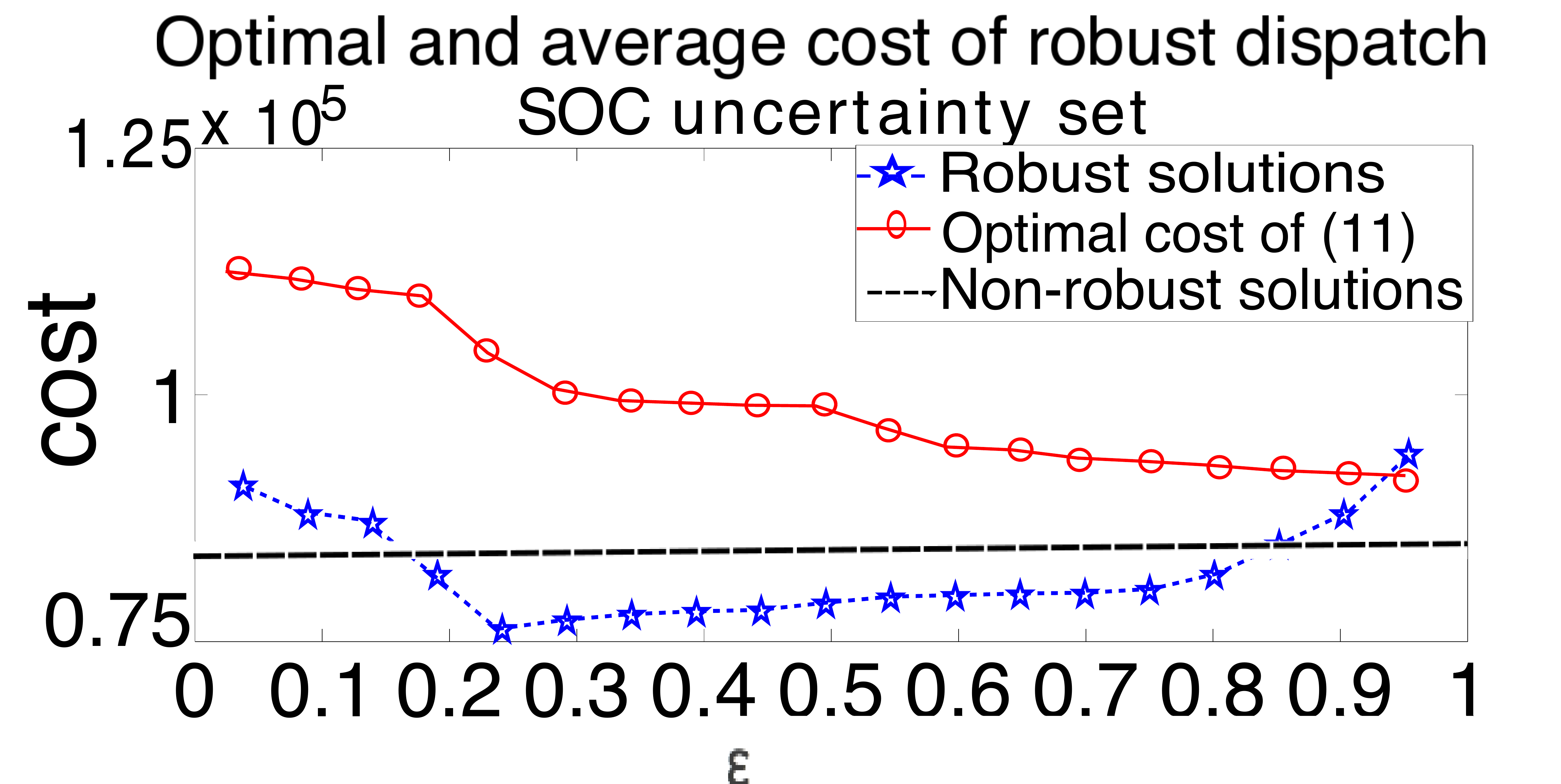}
\vspace{-8pt}
\caption{Comparison of the optimal cost of the robust dispatch problem with SOC type of uncertainty set and the average cost when applying the robust solutions for the test subset of sampled $r_c$. When $\epsilon=0.25$ the average cost is the smallest.}
\label{epsilon_cost_soc}
\vspace{-10pt}
\end{figure}

\textbf{How probabilistic guarantee level affects the average cost}:
There exists a trade-off between the probabilistic guarantee level and the average cost with respect to a random vector $r_c$. Selecting a value for $\epsilon$ is case by case, depending on whether a performance guarantee for the worst case scenario is more important or the average performance is more important. For a high probabilistic guarantee level or a large $1-\epsilon$ value, the average cost may not be good enough since we minimize a worst case that rarely happens in the real world. When $(1-\epsilon)$ is relatively small, the average cost can also be large since many possible values of the random vector are not considered.

We compare the optimal cost of robust solutions and the average cost of empirical tests for two types of uncertainty sets via solving~\eqref{conv_Tm_dep} and~\eqref{u_cs_opt} in Figure~\ref{epsilon_cost} and~\ref{epsilon_cost_soc}, respectively. The optimal cost shows that the result of minimized worst case scenario for all possible $r_c$ included in the uncertainty set, and the average cost shows the of empirical testing cost when we applying the optimal solution to dispatch taxis under random testing data of demand $r_c$. The horizontal line shows the average cost of non-robust solutions that not related to $\epsilon$. The $\epsilon$ values that provide the best average costs are not exactly the same for different types of uncertainty sets according to the experiments. For the box type of uncertainty set in Figure~\ref{epsilon_cost}, $\epsilon=0.3$ provides the smallest average experimental cost; and for SOC type of uncertainty set in Figure~\ref{epsilon_cost_soc}, $\epsilon=0.25$ provides the smallest average cost. The minimum average cost of an SOC robust dispatch solution is smaller than that of a box type. It indicates that the second order moment information of the random variable should be included for modeling the uncertainty set and calculating robust dispatch solutions, though its computational cost is higher.

\section{Conclusion}
\label{sec:conclusion}
In this paper, we develop a multi-stage robust optimization model considering demand model uncertainties in taxi dispatch problems. We model spatial-temporal correlations of the uncertainty demand by partitioning the entire data set according to categorical information, and applying theories without assumptions on the true distribution of the random demand vector. We prove that an equivalent computationally tractable form exist with the constructed polytope and SOC types of uncertainty sets, and the robust taxi dispatch solutions are applicable for a large-scale transportation system. A robust dispatch formulation that purely minimizes the worst-case cost under all possible demand usually sacrifices the average system performance. The robust dispatch method we design allows any probabilistic guarantee level for a minimum cost solution, considering the trade-off between the worst-case cost and the average performance. Evaluations show that under the robust dispatch framework we design, the average demand-supply ratio mismatch error is reduced by $31.7\%$, and the average total idle driving distance is reduced by $10.13\%$ or about $20$ million miles in total in one year. In the future, we will enhance problem formulation considering more uncertain characteristics of taxi network model, like traffic conditions. 

\bibliographystyle{abbrv}
{  \small 
\bibliography{uncertain}
}
\section{Appendix}
\label{appendix}
\subsection{Proof of Lemma~\ref{jematch}}
\label{appendix_probform}
\begin{proof}
We first consider the problem of minimizing $\sum_{i=1}^{n}\left|\frac{r^k_i}{\sum\limits_{j=1}^{n} X^k_{ji}-\sum\limits_{j=1}^{n} X^k_{ij} + L^k_i}-\frac{\sum\limits_{j=1}^n r_j^k}{N^k}\right|$ for one time slot $k$. %To simplify notation, let $r_i^k=a_i \geqslant 1$,
\begin{align}
\sum\limits_{j=1}^{n} X^k_{ji}-\sum\limits_{j=1}^{n} X^k_{ij} + L^k_i =b^k_i, \ i=1,\dots, n.
\label{bi}
\end{align}
Given a vector $L^k$ that satisfies $L^k_i \geqslant 0$, $\sum\limits_{i=1}^k=N^k$, we have $\sum\limits_{i=1}^n b^k_i=\sum\limits_{i=1}^n L^k_i = N^k$, since balancing vacant vehicles does not change the total number of vacant vehicles in the city. 
 
To explain how~\eqref{JE} approximates~\eqref{mismatch} under constraints~\eqref{total_i} and~\eqref{total_N}, consider the following problem given $r^k_1, \dots, r^k_n$, $N^k=c$:
\begin{align}
	\underset{b^k_i>0, \sum_i b^k_i = c}{\text{minimize}}\sum_i \frac{r^k_i}{(b^k_i)^\alpha},\quad c\ \text{is a constant}.
	\label{ab}
\end{align}
We substitute $b^k_n=c-b^k_1\dots-b^k_{n-1}$ into~\eqref{ab}, and take partial derivatives of $\sum_i \frac{r^k_i}{(b^k_i)^\alpha}$ over $b^k_i, i=1,\dots, n-1$. When the minimum of~\eqref{JE} is achieved, each partial derivative should be $0$, $-\alpha \frac{r^k_i}{(b^k)_i^{\alpha+1}}-\alpha (-1)\frac{r^k_n}{(c-b^k_1\dots-b^k_{n-1})^{\alpha+1}}=0$, which is equivalent to $\frac{r^k_1}{(b^k_1)^{\alpha+1}}=\dots=\frac{r^k_{n-1}}{(b^k_{n-1})^{\alpha+1}}=\frac{r^k_n}{(b^k_n)^{\alpha+1}}.$

Let $\frac{r^k_1}{(b^k_1)^{\alpha+1}}=\dots=\frac{r^k_{n-1}}{(b^k_{n-1})^{\alpha+1}}=\frac{r^k_n}{(b^k_n)^{\alpha+1}}=c_0$, $\gamma=\frac{1}{\alpha+1}$, when $\alpha>0$, $0<\gamma<1$. Assume that $\sum\limits_{i=1}^{n}r^k_i = a$, then
\begin{align*}
\begin{split}
(r^k_1)^{\gamma}&=b^k_1 c_0,\quad \dots,\quad(r^k_n)^{\gamma}=b^k_n c_0, \\
\sum\limits_{i=1}^{n}(r^k_i)^{\gamma}&=(b^k_1+\dots+b^k_n)c_0=c c_0 \Rightarrow c_0= \frac{1}{c} \sum\limits_{i=1}^{n}\sum\limits_{i=1}^{n}(r^k_i)^{\gamma},\\
(r^k_i)^{\gamma}&=\frac{b_i}{c} \sum\limits_{j=1}^{n}(r^k_j)^{\gamma},\quad \frac{r^k_i}{b^k_i}=\frac{(a^k_i)^{1-\gamma}}{c} \sum\limits_{j=1}^{n}(r^k_j)^{\gamma}
\end{split}
%\label{aibi}
\end{align*}
We would like to prove that for any $\epsilon_0 >0$,  any $i\in\{1, \dots, n\}$, there exists a $0 < \gamma <1$, such that
\begin{align}
\begin{split}
 \left|\frac{(r^k_i)^{1-\gamma}}{c} \sum\limits_{j=1}^{n}(r^k_j)^{\gamma}-\frac{a}{c} \right| < \epsilon_0,
%\quad\sum\limits_{i=1}^n  \left|\frac{(r^k_i)^{1-\gamma}}{c} \sum\limits_{j=1}^{n}(r^k_j)^{\gamma}-\frac{a}{c} \right| < n \epsilon_0.
\end{split}
\label{epsilon0}
\end{align}

To prove~\eqref{epsilon0}, it is worth noting that for any given values of $r^k_i \geqslant 1,i=1,\dots,n, c>0$, function $f_i(\gamma)=\frac{(r^k_i)^{1-\gamma}}{c} \sum\limits_{j=1}^{n}(r^k_j)^{\gamma}$ is a continuous function of $\gamma$, and $f_i(\gamma=1)=\frac{a}{c}$ for any $i$. Then for any $\epsilon_0>0$ and any $(i, k)$, there exists a $\delta^k_i >0$, such that
\begin{align*}
	|\gamma-1| < \delta^k_i \Rightarrow \left|\frac{(r^k_i)^{1-\gamma}}{c} \sum\limits_{j=1}^{n}(r^k_j)^{\gamma}-\frac{a}{c} \right| < \epsilon_0. 
\end{align*} 
%Considering the case for all time slots $k$, similarly to the above proof, we conclude that for each index $k$ and $i$, there exists a $\delta_i^k>0$, such that $|\gamma-1| < \delta_i$  and the inequality~\eqref{epsilon0} holds with values of $a,b,c$ calculated from the demand $r_i^k$, supply $L_i^k$ and~\eqref{bi} of time $k$. 

Then let $\delta=\text{min}\{\delta_1^1, \delta_2^1,\dots,\delta^k_n\}$ (when $\epsilon_0$ is small, $\delta$ indicates a small range, so $0<\delta <1$), then for any $\gamma$ in the range $1-\delta <\gamma <1$, the inequality~\eqref{epsilon0} holds for all $k$. Without loss of generality, let $\gamma=1-0.5 \delta$, $\alpha=\frac{2}{2-\delta}$, then the optimal solution of problem~\eqref{ab} is $b_i^k=\frac{cr_i^k}{(r^k_i)^{0.5\delta}\sum\limits_{j=1}^{n}(r^k_j)^{1-0.5\delta}}$, $i=1,\dots,n$.

It is worth noting that given any values of $b^k_1>0, \dots, b^k_n>0$, $L^k_1\geqslant 0, \dots, L^k_n\geqslant 0$ that satisfies $\sum\limits_{i=1}^n b^k_i=\sum\limits_{i=1}^n L^k_i$, the equation set~\eqref{bi} has a feasible solution for $n\times n$ variables of the matrix $X^k$. This can be checked by vectorizing matrix $X^k$ to a vector $Y^k \in \mathbb{R}^{n^2}$ and transforming equation set~\eqref{bi} to a new equation set of $Y^k$. We get a homogeneous equation set with $n$ equations and $n \times n$ variables of $Y^k$, which always has a feasible solution. Hence, we plug in the values of $b_i^k=\frac{cr_i^k}{(a^k_i)^{0.5\delta}\sum\limits_{j=1}^{n}(r^k_j)^{1-0.5\delta}}$ to~\eqref{bi} to get values of $X^k_{ij}$. When a solution violates the non-negative constraint of $X_{ij}^k$, just compare the value of $X_{ij}^k$ and $X_{ji}^k$, without loss of generality we assume that $X_{ij}^k > X_{ji}^k$, then let the final feasible solution be $(X_{ij}^k)' = X_{ij}^k-X_{ji}^k$, $(X_{ji}^k)'=0$, the equation set~\eqref{bi} still holds and we have a non-negative optimal solution of $X^k_{ij}, X^k_{ji}$ that keeps the inequality~\eqref{epsilon0_lemma} hold. It is worth noting that we may have multiple optimal solutions of $X^k_{ij}$ by minimizing~\eqref{JE} under constraints~\eqref{total_i} and~\eqref{total_N}, with $\alpha=\frac{2}{2-\delta}$. However, these optimal solutions will result in different values of the other term~\eqref{JD} about the total idle distance in the objective function of~\eqref{Tm}, and only solutions of problem~\eqref{JE} that also satisfy other constraints such as~\eqref{idle_upper} can be feasible solutions of problem~\eqref{Tm}. Hence, we use~\eqref{JE} as an service fairness metric term of the objective function for problem~\eqref{Tm}, and approximately minimize the difference between local and global demand-supply ratios by minimizing~\eqref{JE}.

It is worth noting that when $\epsilon_0$ is small and $\gamma_0$ is close to $1$, $\alpha$ is close to 0.
\end{proof}

\subsection{Proof of Theorem~\ref{T1_convex}}
\label{appendix_T1}
\begin{proof}To find the equivalent form of the minimax problem~\eqref{Tm} when $\tau=1$ (here we only have variable $X$, not $X^2,\dots, X^{\tau}$), the main step is to find the dual problem of the maximization over $r$ for any fixed $X$ and $L$. No constraint of problem~\eqref{Tm} is a function of $r$, when considering the maximization problem with variable $r$ and already fixed $X$ and $L$, the constraints do not affect the values of $r$. Hence, to find the equivalent minimization form of the maximization problem, we do not include constraints irrelevant to $r$ and only consider the objective function part. For any fixed $X$ and $L$, the maximum part of problem~\eqref{Tm} is equivalent to
%\footnotesize
\begin{align}
\begin{split}
\underset{r\in \Delta}{\text{max}}\quad & J_D(X)+\beta J_E(X,r)=J_D(X)+ c^T(X) r\\
[c(X)]_i= &\beta\frac{1}{(\sum\limits_{j=1}^{n} X_{ji}-\sum\limits_{j=1}^{n}X_{ij}+ L_i)^\alpha}, \\
J_D(X)=&\sum_{i} \sum_{j}X_{ij} W_{ij}.
\end{split}
\label{primal}
\end{align}
%\normalsize

Here $J_E(X,r)$ is affine (also concave) of $r$ for any fixed value of $(X,L)$, since with $(X,L)$ fixed, function $[c(X)]_i$ also has a fixed value. And $J_E(X,r)$ a convex function of $(X,L)$ for any fixed value of $r$. The function of power $\frac{1}{x^{\alpha}}$ is convex on scalar $x>0$ when $\alpha>0$~\cite[Chapter 3.1.5]{book_convex}. Consider a concatenated matrix $[X,L]\in \mathbb{R}^{n\times (n+1)}$ with the last column as vector $L\in\mathbb{R}^{n}$, and a matrix $A^i \in \mathbb{R}^{n\times (n+1)}$ with $A^i_{ji}=1$, $j=1,\dots,n$, $A^i_{ij}=-1$, $j=1,\dots,n$, $A^i_{i,(n+1)}=1$. Then $\sum\limits_{j=1}^{n} X_{ji}-\sum\limits_{j=1}^{n}X_{ij}+ L_i=Tr A^i [X,L]=\sum\limits_{i=1}^n\sum\limits_{j=1}^{n+1} A^i_{ij} [X,L]_{ij}$,  and $[c(X)]_i=\frac{1}{(Tr A^i [X,L])^{\alpha}}$ is a composition of function $\frac{1}{x^{\alpha}}$ with affine mapping $Tr A^i [X,L]: \mathbb{R}^{n \times (n+1)} \to R_+$, trace of the multiplication of matrices $[X,L]$ and $A^i$. Because composition with an affine mapping is an operation that preserves convexity~\cite[Chapter 3.2.2]{book_convex}, $[c(X)]_i$ is a convex function of $X$ and $L$. Finally, $J_E(X,r)=\sum\limits_{i=1}^n \beta r_i [c(X)]_i$, $\beta r_i \geqslant 0$ is a nonnegative weighted sum of convex functions $[c(X)]_i$, an operation that preserves convexity~\cite[Chapter 3.2.1]{book_convex}. Hence, $J_E(X,r)$ is a convex function of $X$ and $L$. 

The Lagrangian of problem~\eqref{primal} with the Lagrangian multipliers $\lambda \geq 0, v\geq 0$ is
%\begin{align*}
%\begin{split}
$\mathcal{L}(X,r,\lambda,v)%&=c^T(X)r+d(X)-\lambda^T(Ar-b)+v^Tr\\
=J_D(X)+b^T\lambda-(A^T\lambda-c(X)-v)^T r,$
%\end{split}
%\end{align*}
where $(A^T\lambda-c(X)-v)^T r$ is a linear function of $r$, and the upper bound exists only when $A^T\lambda-c(X)-v=0$.
The objective function of the dual problem is
\begin{align*}
\begin{split}
g(X,\lambda,v)&=\sup_{r\in \Delta} \mathcal{L}(X,r,\lambda,v)\\
                      &=\begin{cases} J_D(X)+b^T\lambda\quad \text{if}\quad A^T\lambda-c(X)-v=0.\\
                                                         \infty \quad\text{otherwise}
                          \end{cases}
\end{split}
%\label{dual_obj}
\end{align*}
With $v\geq 0,$ the constraint $A^T\lambda-c(X)-v=0$ is equivalent to $A^T\lambda-c(X) \geq 0.$ Strong duality holds for problem of~\eqref{primal} since it satisfies the refined Slater's condition for affine inequality constraints~\cite[Chapter 5.2.3]{book_convex}---the primal problem is convex, $c^T(X) r$ is affine of $r$, and by the definition of the uncertainty set, the non-empty affine inequality constraint of $r$ is feasible. The primal convex problem is feasible with affine inequality constraints. The dual problem of~\eqref{primal} is
\begin{align}
\begin{split}
\underset{\lambda \geq 0}{\text{min}} \quad J_D(X)+b^T\lambda \quad
\text{s.t.}\quad A^T\lambda-c(X) \geq 0.
\end{split}
\label{dual_T1}
\end{align}

The minimization problem~\eqref{dual_T1} is  the dual problem of~\eqref{primal} with the same optimal cost for any fixed value of $X$ and $L$, and problem~\eqref{Tm} is to minimize the same objective $J_D(X)+b^T\lambda$ also over $X$ (when $T=1$, $L$ is the number of initial empty vehicles at each region measured by GPS data, so $L$ is a provided parameter in this case. When $\tau \geqslant 2$, $L^k, k=2,\dots, \tau$ are variables) together with the constraints about $X$. The constraint $A^T-c(X)\geq 0$ is convex of $X$, since $[c(X)]_i$ is convex of $X$ for $i=1,\dots,n$, and the constraint of $A^T-c(X)\geq 0$ is equivalent to $n$ inequalities between convex functions and a scalar $0$, which are convex constraint inequalities. Grouping the minimization objective and all the constraints of problem~\eqref{Tm}, we get problem~\eqref{conv_T1} as the equivalent convex optimization form of problem~\eqref{Tm}.
\end{proof}

\subsection{Proof of Lemma~\ref{lemma_minimax}}
\label{appendix_minimax}
\begin{proof}
Now consider the maximin problem over stage $k$ and $k+1$, $1\leqslant k \leqslant \tau-1$ of problem~\eqref{Tm}
\begin{align}
\begin{split}
\underset{r^{k}\in\Delta_k}{\text{max}}\ \underset{X^{k+1},L^{k+1}}{\text{min}}\quad J =&\sum_{k=1}^{\tau} (J_D(X^k)+\beta J_E(X^k,r^k))\\
%=& \sum_{k=1}^{\tau} \sum_{i}\left( \sum_{j} X^k_{ij} W_{ij}+\frac{\beta r^k_i}{(\mathbf{1}_n^T X^k_{\cdot i}-X^k_{i\cdot}\mathbf{1}_n+ L^k_i)^{\alpha}}\right) \\
\text{s.t.}%\quad A_{\tau}^T\lambda-c(X^{\tau})\geq 0,\\
                 \quad &\text{constraints of~\eqref{Tm}}. 
\end{split}
\label{minimax}
\end{align}
The domain of problem~\eqref{minimax} satisfies that $X^{k+1}, L^{k+1}, \lambda$ is compact, and the domain of $r^{k}$ is compact. The objective function is a closed function convex over $X^{k+1}, L^{k+1}$ and concave over $r^{k}$. According to Proposition $2.6.9$ with condition (1) of~\cite{analysis_conv}, when the objective and constraint functions are convex of the decision variables, concave of the uncertain parameters, and the domain of decision variables and uncertain parameters are compact, the set of saddle points for the maximin problem at time $k$ and $k+1$, i.e., $\underset{r^{k}\in\Delta_k}{\text{max}}\ \underset{X^{k+1},L^{k+1}}{\text{min}} J$ with the objective function and constraints of problem~\eqref{minimax} is nonempty. The minimax equality holds for problem~\eqref{minimax} at time $k$ and $k+1$: 
%\begin{align*}
\\\centerline{$\underset{r^{k}\in\Delta_k}{\text{max}}\ \underset{X^{k+1},L^{k+1}}{\text{min}} J =\underset{X^{k+1},L^{k+1}}{\text{min}}\ \underset{r^{k}\in\Delta_k}{\text{max}} J.
$}
%\end{align*}
Repeat the above proof process from $k=\tau-1$ backwards to $k=1$, we get a minimax form of robust optimization problem $\underset{X^{1:\tau},L^{2:\tau}}{\text{min}}\ \underset{r^{1}\in\Delta_1, \dots, r^{\tau} \in \Delta_{\tau}}{\text{max}} J=\underset{X^{1:\tau},L^{2:\tau}}{\text{min}}\ \underset{r_c \in \Delta}{\text{max}} J$. %By the definition of demand uncertainty sets, $\delta_k$ is a projection of $\Delta$, and the minimax problem is equivalent to $\underset{X^{1:\tau},L^{2:\tau}}{\text{min}}\ \underset{r_c \in \Delta}{\text{max}} J$. We then get the conclusion of this lemma.
\end{proof}

\subsection{Proof of Lemma~\ref{lemma_Tm_poly} and Theorem~\ref{Tm_poly}}
\label{appendix_poly}
\subsubsection{Proof of Lemma~\ref{lemma_Tm_poly}}
\label{appendix_Tm_poly1}
\begin{proof}
With the polytope form of uncertainty set~\eqref{polytope}, the domain of each $r^k$ is closed and convex, i.e., is compact, and Lemma~\ref{lemma_minimax} holds. Considering the maximizing part of problem~\eqref{Tm_minimax}
\begin{align}
%\begin{split}
\underset{r^1\in \Delta_1, \dots, r^{\tau}\in\Delta_{\tau}}{\text{max}}\quad  J, \quad
\text{s.t.}    \quad \text{constraints of~\eqref{Tm}},
%\end{split}
\label{obj_Tm}
\end{align}
the Lagrangian of~\eqref{obj_Tm} with multipliers $\lambda^k \geqslant 0, v^k \geqslant 0$ is $\mathcal{L}(X^k,r^k,\lambda^k,v^k)=\sum_{k=1}^{\tau}(J_D(X^k)+b_k^T\lambda^k-(A_k^T\lambda^k-c(X^k)-v^k)^T r^k)$.
%\begin{align}
%\begin{split}
%=&\sum_{k=1}^{\tau}(c^T(X^k)r^k+d(X^k)-(\lambda^k)^T(A_kr^k-b_k)+(v^k)^Tr^k),\\
%\end{split}
%\label{lag_Tm_indep}
%\end{align}
Hence, based on the proof of Theorem~\ref{T1_convex}, we take partial derivative of the Lagrangian for every $r^k\in \Delta_k$. The inequality constraint of $r^k\in\Delta_k$ defined as~\eqref{polytope} is affine of $r^k$ and feasible (non-empty), $c^T(X^k)r^k$ is affine of $r^k$, and problem~\eqref{obj_Tm} is convex with feasible affine inequality constraints. Hence, refined Slater's condition for affine constraints is satisfied and strong duality holds for problem~\eqref{obj_Tm}. An equivalent form of~\eqref{Tm} under uncertainty set~\eqref{polytope} is defined as~\eqref{conv_Tm}.
\end{proof}

\subsubsection{Proof of Theorem~\ref{Tm_poly}}
\label{appendix_Tm_poly2}
\begin{proof}
With uncertain set defined as~\eqref{delta}, the domain of each $r^k$ is compact and Lemma~\ref{lemma_minimax} holds. We consider the equivalent problem~\eqref{Tm_minimax} of~\eqref{Tm}, and first derive the Lagrangian of the maximum part of the objective function~\eqref{Tm_minimax} with constraint $\lambda \geq 0, v_{k} \geq 0$% for the last stage $\tau$ 
\begin{align}
\begin{split}
    &\mathcal{L}(X^k,r^k, \lambda,v_{k})\\
%=&\sum_{k=1}^{\tau} (c^T(X^k)r^k+d(X^k))+v_{\tau}^T r^{\tau}\\
                                              % &-\lambda^T(A_1 r^1+\dots + A_{\tau} r^{\tau}-b),\\
                                           =&b^T\lambda-\sum_{k=1}^{\tau}((A^T_k\lambda-c(X^k)-v_k)^T r^k-J_D(X^k)),%-v_{\tau}^T r^{\tau}, 
\end{split}
\label{lag_Tm}
\end{align}
%where $(A^T\lambda-c(X)-v)^T r$ is a linear function of $r$, and the lower bound exists only when $A^T\lambda-c(X)-v=0$.
Similarly as the proof of Theorem~\ref{T1_convex}, we take the partial derivative of~\eqref{lag_Tm} over each $r^{k}$, the objective function of the dual problem is 
\begin{align*}\begin{split} 
%g(X^k,L^k, \lambda, r^k)=
\sup_{r^k \in \Delta_k} \mathcal{L}(X^k,r^k,\lambda,v_{k})=\sum\limits_{k=1}^{\tau} J_D(X^k)+b^T\lambda 
\end{split}
\end{align*}
when $A_{k}^T\lambda-c(X^{k})-v_{k}=0$.

%&
%=&\begin{cases} \infty \quad\text{if}\ \exists k\ \text{s.t.}\ A_{k}^T\lambda-c(X^{k})-v_{k}\neq0, \\
 %\quad \text{o.w.}                                 
     %                     \end{cases}

Since the inequality constraint of the uncertainty set defined as~\eqref{delta} is affine of each $r^k$ and feasible (non-empty uncertainty set), $c^T(X^k)r^k$ is affine of $r^k$, and problem~\eqref{obj_Tm} is convex with feasible affine inequality constraints, refined Slater's condition with affine inequality constraints is satisfied. Then strong duality holds, problem~\eqref{conv_Tm_dep} is a equivalent to the computationally tractable convex optimization form~\eqref{Tm} under uncertain set~\eqref{delta}.
\end{proof}

\subsection{Proof of Theorem~\ref{theorem_soc}}
\label{appendix_soc}
\begin{proof}
Under the definition of uncertainty set~\eqref{u_cs} for concatenated $r^k$, the domain of each $r^k$ is compact, and problem~\eqref{Tm} is equivalent to~\eqref{Tm_minimax}. We now consider the dual form for the objective function $\sum\limits_{k=1}^{\tau}J_E(X^k, r^k)$ that relates to $r^k$. By the definition of inner product, we have
%\begin{align*}
$\sum_{k=1}^{\tau}c^T(X^k) r^k = c^T_l (X) r_c,  c_l(X)=[c^T(X^1)\ \dots\ c^T(X^\tau)]^T.$
%\end{align*} 
%and $r_c$ is the concatenation of $r^1,\dots, r^\tau$. 
When the uncertainty set of $r_c$ is an SOC defined as~\eqref{u_cs}, problem~\eqref{Tm_minimax} is equivalent to
\begin{align}
\begin{split}
\underset{X^k, L^k}{\text{min}}\ \underset{r_c \geqslant 0}{\text{max}} \quad&\left(c_l^T(X) r_c + \sum_{k=1}^{\tau}\sum_{i} \sum_{j}X^k_{ij} W_{ij}\right)\\
\text{s.t.}\quad & r_c= \hat{r}_c +y+C^Tw,\\
                                     &\|y\|_2 \leqslant \Gamma_1^B, \|w\|_2 \leqslant \sqrt{\frac{1}{\epsilon}-1},\\
\quad&\text{constraints of~\eqref{Tm}.}
\end{split}
\label{max_rc}
\end{align}

We first consider the following minimax problem related to the uncertainty set 
%$max_{r_c \in \mathcal{U}}_{\epsilon}^{CS} min_{z \geq 0, v} L = min max L$
\begin{align}
\begin{split}
\underset{r_c \geqslant 0}{\text{max}} \quad&c_l^T(X) r_c\\
\text{subject to}\quad & r_c= \hat{r}_c +y+C^Tw,\\
%\Gamma_1^B \|z\|_2 +\sqrt{\frac{1}{\epsilon}-1} \|Cz\|_2 \leqslant t,\\  
                                     &\|y\|_2 \leqslant \Gamma_1^B, \|w\|_2 \leqslant \sqrt{\frac{1}{\epsilon}-1}.
%&\|u\|_2 \leqslant \lambda, \lambda \geqslant 0,\\
%&\text{constraints of~\eqref{Tm}},\quad k=1,\dots, \tau,
\end{split}
\label{max_rc}
\end{align}
The constraints of problem~\eqref{max_rc} have a feasible solution $r_c=\hat{r}_c$, $y=0$ and $w=0$, such that $\|y\|_2 < \Gamma_1^B, \|w\|_2 < \sqrt{\frac{1}{\epsilon}-1}$, and $c_l^T(X) r_c$ is affine of $r_c$. Hence, Slater's condition is satisfied and strong duality holds.

To get the dual form of problem~\eqref{max_rc}, we start from the following Lagrangian with $v\geqslant 0$,
%\footnotesize
%\begin{align*}
$\mathcal{L}(X,r_c, z, v)
=c_l^T(X) r_c + z^T (\hat{r}_c +y+C^Tw-r_c)+v^T r_c.$
%\end{align*}
%\normalsize
By taking the partial derivative of the above Lagrangian over $r_c$, we get the supreme value of the Lagrangian as
\begin{align*}
\sup_{r_c} \mathcal{L}(X,r_c, z, v)=&\begin{cases} z^T (\hat{r}_c +y+C^Tw)\quad \text{if}\ c_l(X) \leqslant z\\  
           \infty \quad \text{o.w.}
\end{cases}
\end{align*}
Then with the norm bound of $y$ and $w$, we have
\begin{align*}
&\sup_{\|y\|_2 \leqslant \Gamma_1^B, \|w\|_2 \leqslant \sqrt{\frac{1}{\epsilon}-1}}(z^T (\hat{r}_c +y+C^Tw))\\
=&\hat{r}^T_c z+\Gamma_1^B \|z\|_2 +\sqrt{\frac{1}{\epsilon}-1} \|Cz\|_2.
\end{align*} 
Hence, the objective function of the dual problem for~\eqref{max_rc} is
%$y$ and $w$ respectively, the objective function of the dual problem is 
\begin{align*}
\begin{split}
&g(X, r_c, z)=\sup_{r_c \in \mathcal{U}_{\epsilon}^{CS}} \mathcal{L}(X,r_c, z)\\
=&\begin{cases} 
 \hat{r}^T_c z+\Gamma_1^B \|z\|_2 +\sqrt{\frac{1}{\epsilon}-1} \|Cz\|_2,\quad
  \text{if}\ c_l(X) \leqslant z \\
           \infty \quad \text{o.w.}.\\                                
\end{cases}
\end{split}
\end{align*}
Together with the objective function $J_D(X^k)$ and other constraints that do not directly involve $r_c$, an equivalent convex form of~\eqref{Tm} given the uncertainty set~\eqref{u_cs} is shown as~\eqref{u_cs_opt}.
\end{proof}
\end{document}